\documentclass[11pt,leqno]{amsart}
\usepackage{a4wide,color,array}
\usepackage{cite}
\usepackage{hyperref}
\usepackage{amsmath,amssymb,amsthm,color}
\hypersetup{linktocpage,colorlinks, citecolor={magenta}}
\usepackage[english]{babel}
\usepackage[utf8]{inputenc}
\usepackage{enumitem}
\setlist{nosep} % or \setlist{noitemsep} to leave space around whole list
\usepackage{mathtools}
\makeatletter
\def\namedlabel#1#2{\begingroup
    #2%
    \def\@currentlabel{#2}%
    \phantomsection\label{#1}\endgroup
}
\makeatother
\newcommand{\ed}[1]{ #1 }

\newcommand{\edb}[1]{ #1 }

\newcommand{\Z}{\mathbb{Z}}
\newcommand{\R}{\mathbb{R}}

\newtheorem{theo}{Theorem}[section]
\newtheorem{prop}[theo]{Proposition}
\newtheorem{lem}[theo]{Lemma}
\newtheorem{coro}[theo]{Corollary}

\theoremstyle{remark}

\theoremstyle{plain}

\numberwithin{equation}{section}

\def\t0{\rightarrow 0}
\def\ti{\rightarrow \infty}

\newcommand{\f}{\frac}

\newcommand{\hal}{\frac{1}{2}}

\def\1{\mathbf{1}} % Fonction caract"ristique

\def \mc{\mathcal}

\renewcommand{\epsilon}{\varepsilon}

\def\Rd{\R^d} % Espace physique
\def \ZNbeta{Z_{N,\beta}} % Fonction de partition
\def \carr{K} % Carr"

\def \om{\overline{m}} % Bornes sur la densit"
\def\({\left(}
\def\){\right)}

\def\config{\mathcal{X}} % Espace des configurations de points
\def\dconfig{d_{\config}}

\def \probas{\mathcal{M}}
\def\P{\mathbb{P}} % "Vraies" mesures
\def \Pst{P} % Processus ponctuel standard sans marque
\def \PNbeta{\P_{N, \beta}} % Mesure de Gibbs " \beta
\def \PgN2{\mathbf{P}_{N,2}} % Processus Gibbsien " N points et " \beta = 2..
\def \HN{\mathcal{H}_N}

\def\Esp{\mathbf{E}} % Esp"rance
\def \E{\Esp}

\def \Ent{\mathrm{Ent}}   % Entropie relative classique
\def \ERS{\mathsf{ent}} % Entropie relative sp"cifique (sans marques)
\def \Leb{\mathbf{Leb}}
\def \Poisson{\mathbf{\Pi}}
\def \fbarbeta{\overline{\mathcal{F}}_{\beta}} % Rate function avec marque
\def \dist{d}

\def \dist{\mathrm{dist}}

\def\XXint#1#2#3{{\setbox0=\hbox{$#1{#2#3}{\int}$}
     \vcenter{\hbox{$#2#3$}}\kern-.5\wd0}}

\def \Int{\mathrm{Int}}

\def\Cabs{\mathcal{C}}

\newcommand{\crd}[1]{|#1|}

\def \XN{\vec{X}_N}
\def \C{\mathcal{C}}
\def \Ac{\Omega}

\def \config{\mathcal{X}} %% Config
\def \dconfig{d_{\config}} %% Distance
\def \bconfig{\overline{\config}} %% Tagged config
\def \dbconfig{d_{\bconfig}} %% Distance tagged
\def \pconfig{\mathcal{P}(\config)} %% Random point config
\def \P{P} %% typical random point config
\def \psconfig{\mathcal{P}_{stat}(\config)} %% Stationary
\def \psmconfig{\mathcal{P}_{stat,m}(\config)} %% 
\def \psunconfig{\mathcal{P}_{stat,1}(\config)} %% 
\def \pbconfig{\overline{\mathcal{M}}(\bconfig)} %% Tagged random
\def \bP{\overline{P}} %% typical random tagged
\def \bPx{\overline{P}^x} %% disintegration
\def \psbconfig{\overline{\mathcal{M}}_{stat}(\bconfig)} %% Stationary tagged
\def \psunbconfig{\overline{\mathcal{M}}_{stat,1}(\bconfig)}
\def \om{\omega}
\def \Emp{\mathrm{Emp}}
\def \bEmp{\overline{\Emp}}
\def \Dens{\mathrm{Dens}}
\def \Intens{\mathrm{Intens}}
\def \bIntens{\overline{\mathrm{Intens}}}
\def \Int{\mathrm{Int}}
\def \bsigrho{\overline{\sigma}_{\rho}}

\def \Es{E_s}

\def \Esl{E_{s, \Lambda}}
\def \EslN{\mathcal{E}_{s,\Lambda}(N)}
\def \Csd{C_{s,d}}
\def \Ws{\mathcal{W}_s}
\def \WsP{\mathbb{W}_s}
\def \bWsP{\overline{\mathbb{W}}_s}

\def \fbarbeta{\overline{\mathcal{F}}_{\beta}}
\def \Oml{\mathbf{D}_{\Lambda}}

\def \PNbeta{\mathbb{P}_{N,\beta}}
\def \bQN{\bar{\mathfrak{Q}}_N}

\def \fPNbeta{\overline{\mathfrak{P}}_{N, \beta}}
\def \bV{\overline{\mathbb{V}}}

\def \tK{\tilde{K}}

\def \bPd{\overline{P}_d}
\def \bPc{\overline{P}_c}

\def \Rp{R'}
\def \Aabs{\mathcal{A}}
\def \Ibeta{I_{\beta}}
\def \tIntens{\widetilde{\mathrm{Intens}}}
\def \emp{\mathrm{emp}}

\def \femp{\mathfrak{emp}}

\def \muVb{\mu_{V, \beta}}

\def \brho{\bar{\rho}}
\def \rhounif{\rho_{\mathrm{unif}}}

\def \Leba{\Leb_{\Ac^N}}

\title{Large Deviation Principles for Hypersingular Riesz Gases}

\author{Douglas P. Hardin, Thomas Lebl\'e, Edward B. Saff, Sylvia Serfaty}
\address[Douglas P. Hardin]{Center for Constructive Approximation,  Department of Mathematics,  Vanderbilt University, Nashville, TN, 37240, USA}
\email{doug.hardin@vanderbilt.edu}
\address[Thomas Lebl\'e]{Courant Institute of Mathematical Sciences, 251 Mercer Street, New York University,
New York, NY 10012-1110, USA}
\email{thomasl@math.nyu.edu}
\address[Edward B. Saff]{Center for Constructive Approximation,  Department of Mathematics,  Vanderbilt University, Nashville, TN, 37240, USA}
\email{edward.b.saff@vanderbilt.edu}
\address[Sylvia Serfaty]{Courant Institute of Mathematical Sciences\\ 251 Mercer Street, New York University\\
New York, NY 10012-1110, USA  \&  Institut Universitaire de France  \& Sorbonne Universit\'es, UPMC Univ. Paris 06, CNRS, UMR 7598, Laboratoire Jacques-Louis Lions, 4, place Jussieu 75005, Paris, France.}
\email{serfaty@cims.nyu.edu}
\thanks{The research of the first and third authors was supported, in part,
by the U. S. National Science Foundation under the grant DMS-1516400 and was facilitated by the hospitality and support of  
 the Laboratoire Jacques-Louis Lions at UPMC}

\begin{document}
\date{\today}
\maketitle 


\begin{abstract}
We study $N$-particle systems in $\R^d$ whose interactions are governed by a \edb{hypersingular Riesz potential $|x-y|^{-s}$, $s>d$,  and subject to  an external field.}  We provide both macroscopic results as well as microscopic results in the limit as $N\to \infty$
for random point configurations with respect to the associated Gibbs measure at \ed{scaled inverse} temperature $\beta$.  
We show that a large deviation principle holds with a rate function of the form `$\beta$-Energy +Entropy', yielding that the microscopic
behavior (on the scale $N^{-1/d}$) of such $N$-point systems is asymptotically determined by the minimizers of this rate function.  In 
contrast to the asymptotic behavior in the integrable case $s<d$, where on the macroscopic scale $N$-point empirical measures have
limiting density independent of $\beta$, the limiting density for $s>d$ is strongly $\beta$-dependent.  
\end{abstract}


%\tableofcontents
\section{Introduction and main results}
\subsection{Hypersingular Riesz gases} \label{sec-gensetting}
Let $d \geq 1$ and $s$ be a real number with $s > d$. We consider a system of $N$ points in the Euclidean space $\Rd$ with \textit{hypersingular} Riesz pairwise interactions, in an external field $V$. The particles are assumed to live in a \textit{confinement set} $\Ac \subseteq \R^d$. The energy $\HN(\XN)$ of the system in a given state $\XN = (x_1, \dots, x_N)  \in (\R^d)^N$ is defined to be
\begin{equation} \label{def:HN}
\HN(\XN) := \sum_{1 \leq i \neq j \leq N} \frac{1}{|x_i-x_j|^s} + N^{s/d} \sum_{i=1}^N V(x_i).
\end{equation}
The external field $V$ is a confining potential, growing at infinity, on which we shall make assumptions later. The term \textit{hypersingular} corresponds to the fact that the singularity of the kernel $|x-y|^{-s}$ is non-integrable with respect to the Lebesgue measure on $\Rd$.

For any $\beta > 0$, the canonical Gibbs measure associated to \eqref{def:HN} at inverse temperature $\beta$ and for particles living on $\Ac$ is given by
\begin{equation}\label{def:PNbeta}
d\PNbeta(\XN) = \frac{1}{\ZNbeta} \exp\left( - \beta N^{-s/d} \HN(\XN)\right) \mathbf{1}_{\Ac^N}(\XN) d\XN,
\end{equation}
where $d\XN$ is the Lebesgue measure on $(\R^d)^N$, $\mathbf{1}_{\Ac^N}(\XN)$ is the indicatrix function of $\Ac^N$, and $\ZNbeta$ is the “partition function”; i.e., the normalizing factor  
\begin{equation}
\label{def:ZNbeta} \ZNbeta := \int_{\Ac^N} \exp\left( - \beta N^{-s/d} \HN(\XN)\right) d\XN.
\end{equation}

We will call the statistical physics system described by \eqref{def:HN} and \eqref{def:PNbeta} a \medskip \textit{hypersingular Riesz gas}.

\edb{For  Riesz potentials in the case $s>d$, ground state configurations (or Riesz energy minimizers) of $N$-particle systems  (with or without the external field $V$)  have been extensively studied in the large $N$ limit, see \cite{HSAdv, HSNotices, Hardin:2016kq} and the references therein.
Furthermore, for the case of positive temperature, the statistical mechanics of Riesz gases have been investigated in \cite{LebSer} but for a different range of the parameter $s$, namely $\max(d-2, 0) \leq s < d$.  In that paper,  a large deviation principle for the empirical process (which encodes the microsopic behavior of the particles at scale $N^{-1/d}$, averaged in a certain way) was derived. The main goal of the present paper is to extend that work to the hypersingular case.  By combining the approaches of the above mentioned papers we obtain a large deviation principle describing  macroscopic as well as microscopic properties for hypersingular Riesz gases.}

\ed{
Studying Riesz interactions for the whole range of $s$ from $0$ to infinity is of interest in approximation theory and coding theory, as it connects logarithmic \edb{interactions}, Coulomb  interactions, and \edb{(in the limit $s \to \infty$)} packing problems, see \cite{HSNotices,saff1997distributing}.}\edb{ Investigating}  \ed{such systems with temperature is also a natural question for statistical mechanics, as it improves our understanding of the behavior of systems with long-range vs. short-range interactions (see, for instance, \cite{draw,cdr,MR2673930}  where the interest of such questions is stressed and \cite{bloom2016large} and  \cite[Section 4.2]{mazars2011long} for \edb{additional} results).} Analyzing the case $s > d$ \edb{is also} a first step toward the study of physically more relevant interactions as the Lennard-Jones \medskip potential.

\ed{The hypersingular \edb{Riesz} case $s>d$ and the \edb{integrable Riesz} case $s<d$ have important differences.}
\edb{For} $s < d$ (which can be thought of as long-range) and, more generally, \edb{for integrable interaction kernels $g$} (which includes regular interactions) 
 the global, macroscopic behavior can be studied using classical potential theory\edb{. Namely,} the empirical measure $\frac{1}{N} \sum_{i=1}^N \delta_{x_i}$ is found to converge rapidly to some equilibrium measure determined uniquely by $\Ac$ and $V$ and obtained as the unique minimizer of the potential-theoretic functional
\begin{equation*}
\iint_{\R^d \times \R^d} g(x-y) d\mu(x) d\mu(y) + \int_{\R^d} V d\mu
\end{equation*}which can be seen as a mean-field energy with a non-local term.
We refer e.g. to \cite{safftotik} or \cite[Chap. 2]{MR3309890} for a treatment of this question (among others). 

\ed{In these \edb{integrable} cases, if temperature is scaled in the same way as here, 
the macroscopic behavior is governed by the equilibrium measure \edb{ and thus} is independent of the temperature \edb{so that} no knowledge of the microscopic distribution of points is necessary to determine the macroscopic distribution.  At the next order in energy, which governs the microscopic distribution of the points, a dependency on $\beta$ appears. As seen in \cite{LebSer}, in  the Coulomb and potential Riesz cases (it is important in the method that the interaction kernel be reducible to the kernel of a local operator, which is known only for these particular interactions),  the microscopic distribution around a point is given by a problem in a whole space with a neutralizing background, fixing the local density as equal to that of the equilibrium measure at that point.  The microscopic distribution is found to minimize the sum of a (renormalized) Riesz energy term and a relative entropy term.  A crucial ingredient in the proof is a ``screening" construction \edb{showing that  energy can be computed additively over large disjoint microscopic boxes; i.e., interactions between configurations in different large microscopic boxes are negligible to this order. }}

\ed{The hypersingular case can be seen as more delicate \edb{than} the \edb{integrable} case  due to the absence of an equilibrium measure. The limit of the empirical measure has to be identified differently. In the case of \edb{ground state configurations} (minimizers), this was done in \cite{Hardin:2016kq}. \edb{For positive temperature}, in contrast with the \edb {above described integrable case}, \edb{we shall show} the \edb{empirical } limit \edb{measure} is obtained as a by-product of the study at the microscopic scale and depends  on $\beta $ in quite an indirect way (see Theorem \ref{theo:LDPmesure}). The microscopic profiles minimize a full-space version of the problem, giving an energy that depends on the local density,  and the macroscopic distribution   can then be found by a local density approximation, by minimizing the sum of its energy and that due to the confinement \edb{potential}. Since the energy is easily seen to scale like $N^{1+s/d}$, the choice of the temperature scaling $\beta N^{-s/d}$ is made so that the energy and the entropy for the microscopic distributions carry the same weight of order $N$. Other choices of temperature scalings  are possible, but would lead to degenerate versions of the situation we are examining, with either all the entropy terms asymptotically disappearing for small temperatures, or the effect of the energy altogether disappearing for large temperatures.
  Note that going to the microscopic scale in order to derive the behavior at the macroscopic scale was already the approach needed in \cite{leble2015large} for the case of the ``two-component plasma", a system of two-dimensional particles of positive and negative charges interacting logarithmically for which no a priori knowledge of the equilibrium measure can be found.
 
 On the other hand, the hypersingular case  is also easier in the sense that the interactions decay faster at infinity, \edb{implying that}  long-range interactions between large microscopic ``boxes" \edb{are negligible and do not require any sophisticated screening procedures}. Our proofs will make crucial use of this ``self-screening" property.}

\ed{To describe the system at the microscopic scale, } we define a Riesz energy $\bWsP$ \edb{(see subsection~\ref{sec:energyerandompoint})} for infinite random point configurations which is the counterpart of the renormalized energy of \cite{petrache2014next,LebSer,leble2016logarithmic} (defined for $s < d$). It  is \edb{conjectured} to be minimized by lattices \ed{ \edb{for certain  low dimensions, but this is a completely open problem  with the exception of  dimension 1 (see \cite{blanc2015crystallization} and the discussion following \eqref{perEnLim}).}}

 To any sequence of configurations $\{\XN\}_N$, we associate an ``empirical process" whose limit  (a random tagged point process) describes the point configurations $\XN$ at scale $N^{-1/d}$. Our main result will be that there is a Large Deviations Principle for the law of this empirical process with rate function equal to (a variant of) the energy  $\beta \bWsP$ plus the relative entropy of the empirical process with respect to the Poisson point process.
 
 For minimizers of the Riesz energy $\HN$, we show that the limiting empirical processes must minimize $\bWsP$, thus describing their  microscopic \medskip structure.

\ed{The question of treating more general interactions than the Riesz ones remains widely  open. The fact that the interaction has a precise homogeneity under rescaling is crucial for the hypersingular case treated here. On the other hand, in the \edb{integrable} case, we do not know how to circumvent the need for expressing the energy via the potential generated by the points; i.e., the need  for the Caffarelli-Silvestre representation  of the interaction as the kernel of a local operator (achieved by adding a space dimension).
 }

\subsection{Assumptions and notation}
\label{sec:assumptions}

\subsubsection{Assumptions}
In the rest of the paper, we assume that $\Ac\subset \R^d$ is closed with positive $d$-dimensional Lebesgue measure and that
\begin{align} \label{ass:regAc}
&\text{$\partial \Ac$ is $C^1$,} \\
\label{ass:regV} & \text{$V$ is a continuous,  non-negative real valued function on $\Ac$.}
\end{align}
Furthermore if $\Ac$ is unbounded, we assume that 
\begin{align} \label{ass:croissanceV}
& \lim_{|x| \ti} V(x) = + \infty, \\
\label{ass:integr} & \exists M > 0  \text { such that } \int \exp\left(- M V(x)\right) dx < + \infty.
\end{align}

The assumption \eqref{ass:regAc} on the regularity of $\partial \Ac$ is mostly technical and we believe that it could be relaxed to e.g. $\partial \Ac$ is locally the graph of, say, a Hölder function in $\R^d$. However it is unclear to us what the minimal assumption could be (e.g., is it enough to assume that $\partial \Ac$ has zero measure?). An interesting direction would be to study the case where $\Ac$ is a $p$-rectifiable set in $\R^d$ with $p < d$ (see e.g. \cite{Borodachov:2016kx, Hardin:2016kq}). 

Assumption \eqref{ass:regV} is quite mild (in comparison e.g. with the corresponding assumption in the $s < d$ case, where one wants to ensure some regularity of the so-called equilibrium measure, which is essentially two orders lower than that for $V$) and we believe it to be sharp for our purposes. Assumption \eqref{ass:croissanceV} is an additional confinement assumption, and \eqref{ass:integr} ensures that the partition function $\ZNbeta$, defined in \eqref{def:ZNbeta}, is finite (at least for $N$ large enough). Indeed the interaction energy is non-negative, hence for $N$ large enough \eqref{ass:integr} ensures that the integral defining the partition function is convergent.

\subsubsection{General notation}

We let $\config$ be the space of point configurations in $\R^d$ (see Section \ref{sec:pointconfig} for a precise definition). If $X$ is some measurable space and $x \in X$ we denote by $\delta_x$ the Dirac mass at $x$.

\subsubsection{Empirical measure and empirical \ed{processes}}
Let $\XN = (x_1, \dots, x_N)$ in $\Ac^N$ be fixed. 

\begin{itemize}
\item We define the empirical measure $\emp(\XN)$ as
\begin{equation} \label{def:emp}
\emp(\XN) := \frac{1}{N} \sum_{i=1}^N \delta_{x_i}.
\end{equation}
It is a probability measure on $\Ac$.

\item We define $\XN'$ as the finite configuration rescaled by a factor $N^{1/d}$ 
\begin{equation} \label{def:ompN}
\XN' := \sum_{i=1}^N \delta_{N^{1/d} x_i}.
\end{equation}
It is a point configuration (an element of $\config$), which represents the $N$-tuple of particles $\XN$ seen at microscopic scale.

\item We define the \textit{tagged empirical process} $\bEmp_N(\XN)$ as
\begin{equation}
\label{def:Emp}
\bEmp_N(\XN) :=  \int_{\Ac} \delta_{\left(x,\,  \theta_{N^{1/d}x} \cdot \XN' \right)} dx,
\end{equation}
where $\theta_x$ denotes the translation by $- x$.  It is a positive measure on $\Ac \times \config$. \\
\end{itemize}

Let us now briefly explain the meaning of the last definition \eqref{def:Emp}. 
For any $x \in \Ac$, $\theta_{N^{1/d}x} \cdot \XN'$ is an element of $\config$ which represents the $N$-tuple of particles $\XN$ centered at $x$ and seen at microscopic scale (or, equivalently, seen at microscopic scale and then centered at $N^{1/d} x$). In particular any information about this point configuration in a given ball (around the origin) translates to an information about $\XN'$ around $x$. We may thus think of $\theta_{N^{1/d}x} \cdot \XN'$ as encoding the behavior of $\XN'$ around $x$.

The measure 
\begin{equation} \label{empiricalfield}
\int_{\Ac} \delta_{\theta_{N^{1/d}x} \cdot \XN'} dx
\end{equation}
is a measure on $\config$ which encodes the behaviour of $\XN'$ around each point $x \in \Ac$. We may think of it as the “averaged” microscopic behavior (although it is not, in general, a probability measure, and its mass can be infinite). The measure defined by \eqref{empiricalfield} would correspond to what is called the “empirical field”.

The tagged empirical process $\bEmp_N(\XN)$ is a finer object, because for each $x \in \Ac$ we keep track of the centering point $x$ as well as of the microscopic information $\theta_{N^{1/d}x} \cdot \XN'$ around $x$. It yields a measure on $\Ac \times \config$ whose first marginal is the Lebesgue measure on $\Ac$ and whose second marginal is the (non-tagged) empirical process defined above in \eqref{empiricalfield}. Keeping track of this additional information allows one to test $\bEmp_N(\XN)$ against functions $F(x, \C) \in C^0(\Ac \times \config)$ which may be of the form
$$
F(x, \C) = \chi(x) \tilde{F}(\C),
$$
where $\chi$ is a smooth function localized in a small neighborhood of a given point of $\Ac$, and $\tilde{F}(\C)$ is a continuous function on the space of point configurations. Using such test functions, we may thus study the microsopic behavior of the system after a small average (on a small domain of $\Ac$), whereas the empirical process only allows one to study the microscopic behavior after averaging over the whole $\Ac$.

\ed{The study of empirical processes, or \textit{empirical fields}, as natural quantities to encode the averaged microscopic behavior appear e.g. in \cite{FollmerOrey} for particles without interaction or \cite{Georgii1} in the interacting case.}

\subsubsection{Large deviations principle}
Let us recall that a sequence $\{\mu_N\}_N$ of probability measures on a metric space $X$ is said to satisfy a Large Deviation Principle (LDP) at speed $r_N$ with rate function $I : X \to [0, +\infty]$ if the following holds for any Borel set $A \subset X$
$$
- \inf_{\mathring{A}} I \leq \liminf_{N \ti}\frac{1}{r_N} \log \mu_N(A) \leq \limsup_{N \ti}\frac{1}{r_N} \log \mu_N(A) \leq - \inf_{\bar{A}} I,
$$
where $\mathring{A}$ (resp. $\bar{A}$) denotes the interior (resp. the closure) of $A$. The functional $I$ is said to be a {\it good rate function} if it is lower semi-continuous and has compact sub-level sets. We refer to \cite{MR2571413} and \cite{Varadhan2016} for   detailed treatments of the theory of large deviations and to \cite{MR3309619} for an introduction to the applications of LDP's in the statistical physics setting.

Roughly speaking, a LDP at speed $r_N$ with rate function $I$ expresses the following fact: the probability measures $\mu_N$ concentrate around the points where $I$ vanishes, and any point $x \in X$ such that $I(x) > 0$ is not “seen” with probability $1 - \exp(-N I(x))$.  

\subsection{Main results}
\subsubsection{Large deviations of the empirical processes}
We let $\fPNbeta$ be the push-forward of the Gibbs measure $\PNbeta$ (defined in \eqref{def:PNbeta}) by the map $\bEmp_N$ defined in \eqref{def:Emp}. In other words, $\fPNbeta$ is the law of the random variable “tagged empirical process” when $\XN$ is distributed following $\PNbeta$.

The following theorem, which is the main result of this paper, involves the functional $\fbarbeta=\overline{\mathcal{F}}_{\beta,s}$ defined in \eqref{def:fbarbeta}. It is a free energy functional of the type “$\beta$ Energy + Entropy” (see Section \ref{sec:ERS}, \ref{sec:energy} and \ref{sec:ratefunction} for precise definitions).  The theorem expresses the fact  that the microscopic behavior of the system of particles is determined by the minimization of the functional $\fbarbeta$ and that configurations $\XN$ having empirical processes $\bEmp(\XN)$  far from a minimizer of $\fbarbeta$, have negligible probability of order $\exp(-N)$.
\begin{theo} \label{theo:LDPemp}
For any $\beta > 0$, the  sequence $\{\fPNbeta\}_N$ satisfies a large deviation principle at speed $N$ with good rate function $\fbarbeta - \min \fbarbeta$.
\end{theo}

\begin{coro} \label{coro:ZNbeta}
The first-order expansion of $\log \ZNbeta$ as $N\to \infty$ is
\begin{equation*}
\log \ZNbeta  = - N \min \fbarbeta  + o(N).
\end{equation*}
\end{coro}

\subsubsection{Large deviations of the empirical measure}
As a byproduct of our microscopic study, we derive a large deviation principle which governs the asymptotics of the empirical measure (which is a macroscropic quantity). Let us denote by $\femp_{N,\beta}$ the law of the random variable $\emp(\XN)$ when $\XN$ is distributed according to $\PNbeta$. The rate function $\Ibeta=I_{\beta,s}$, defined in Section \ref{sec:ratefunction} (see \eqref{def:Ibeta}), \ed{has the  form 
\begin{equation}\label{formeI}
\Ibeta(\rho)= \int_{\Omega} f_\beta(\rho(x)) \rho(x)\, dx + \beta \int_{\Omega} V(x) \rho(x) \, dx + \int_{\Omega} \rho(x) \log \rho(x)\, dx,
 \end{equation}
 and is a local density approximation. The function $f_\beta$  in this expression is determined by a minimization problem over the {\it microscopic} empirical processes.}

\begin{theo} \label{theo:LDPmesure}
For any $\beta > 0$, the sequence $\{\femp_{N,\beta}\}_{N}$ obeys a large deviation principle at speed $N$ with good rate function $\Ibeta - \min \Ibeta$. In particular, the empirical measure converges almost surely to the unique minimizer of $\Ibeta$.
\end{theo}

The rate function $\Ibeta$ is quite complicated to study in general. However, \ed{thanks to the convexity of $f_\beta$ and elementary properties of the standard entropy we may characterize its minimizer in some particular cases (see Section \ref{sec:addproofs} for the proof)}:
\begin{prop} \label{prop:muVbeta}
Let $\muVb$ be the unique minimizer of $\Ibeta$.

\begin{enumerate}
\item If $V = 0$ and $\Ac$ is bounded, then $\muVb$ is the uniform probability measure on $\Ac$ for any $\beta > 0$.
\item If $V$ is arbitrary and $\Ac$ is bounded, $\muVb$ converges to the uniform probability measure on $\Ac$ as $\beta \to 0$.
\item If $V$ is arbitrary, $\muVb$ converges to $\mu_{V,\infty}$ as $\beta \to + \infty$, where $\mu_{V, \infty}$ is the limit empirical measure for energy minimizers as defined in the paragraph below.
\end{enumerate}
\end{prop}

\subsubsection{The case of minimizers}
Our remaining results deal with energy minimizers (in statistical physics, this corresponds to setting $\beta = + \infty$). Let $\{\XN\}_N$ be a sequence of point configurations in $\Ac$ such that for any $N \geq 1$, $\XN$ has $N$ points and minimizes $\HN$ on $\Ac^N$. 

The macroscopic behavior is known from \cite{Hardin:2016kq}: there is a unique minimizer $\mu_{V, \infty}$ (the notation differs from \cite{Hardin:2016kq}) of the functional
  \begin{equation} \label{minimiserho}
  \Csd \int_{\Ac} \rho(x)^{1+s/d}\, dx+ \int_{\Ac} V(x) \rho(x)\, dx
  \end{equation}
   among probability densities $\rho$ over $\Ac$ ($\Csd$ is a constant depending on $s,d$ defined in \eqref{def:Csd1}), and the empirical measure $\emp(\XN)$ converges to $\mu_{V, \infty}$ as $N \ti$.   See \eqref{muVinfty} for an explicit formula for $\mu_{V, \infty}$. \ed{Note that the formula \eqref{minimiserho} is what one obtains when letting formally $\beta \to \infty$ in the definition of $\Ibeta$, and is resembling some of the terms arising in  Thomas-Fermi theory (cf. \cite{MR2583992} and \cite{Lieb-TF-Rev}).
 }

The notation for the next statement is given in Sections \ref{sec:pointconfig} and \ref{sec:energy}. Let us simply say that $\bWsP$ (resp. $\Ws$) is an energy functional defined for a random point configuration (resp. a point configuration), and that $\psunbconfig$ (resp. $\config_{\mu_{V,\infty}(x)}$) is some particular subset of random point configurations (resp. of point configurations in $\R^d$). The intensity measure of a random tagged point configuration is defined in Section \ref{sec:intensitymeasure}.
 \begin{prop} \label{prop:minimizers} We have:
\begin{enumerate}
\item  $\{\bEmp(\XN)\}_N$ converges weakly (up to extraction of a subsequence)  to some minimizer $\bP$ of $\bWsP$ over $\psunbconfig$. 
\item The intensity measure of $\bP$ coincides with $\mu_{V, \infty}$. 
\item For $\bP$-almost every $(x, \C)$, the point configuration $\C$ minimizes $\Ws(\C)$ within the class $\config_{\mu_{V,\infty}(x)}$.
\end{enumerate}
 \end{prop}
 The first point expresses the fact that the tagged empirical processes associated to minimizers  converge to minimizers of the “infinite-volume” energy functional $\bWsP$. The second point is a rephrasing of the global result cited above, to which the third point adds some microscopic information.
 
The problem of minimizing the energy functionals $\bWsP$, $\WsP$ or $\Ws$ is hard in general. In dimension $1$, however, it is not too difficult to show that the “crystallization conjecture” holds, namely that the microscopic structure of minimizers is ordered and converge to a lattice:
\begin{prop} \label{prop:crystallization1d}
Assume $d=1$. The unique stationary minimizer of $\WsP$ is the law of $u + \Z$, where $u$ is a uniform choice of the origin in $[0,1]$.
\end{prop}
In dimension $2$, it would be expected that minimizers are given by the triangular (or Abrikosov) lattice, we refer to \cite{blanc2015crystallization} for a recent review of such conjectures. \ed{In large dimensions, it is not expected that  lattices are minimizers.}

\ed{
\subsection{Outline of the method}
Our LDP result is phrased in terms of the empirical processes associated to point configurations, as in \cite{LebSer} and thus the objects we consider and the overall aproach are quite similar to \cite{LebSer}. It is however quite simplified by the fact that, because the interaction is short-range and we are in the non-potential case, we do not need to express the energy in terms of the ``electric potential" generated by the point configuration. The definition of the limiting microscopic interaction energy \edb{$\mathcal W_s(\mathcal C)$} is thus significantly simpler than in \cite{LebSer}, it suffices to take, for $\mathcal{C}$ an infinite configuration of points in the  whole space,
 $$\mathcal W_s(\mathcal C)= \liminf_{R\to \infty} \sum_{p,q \in \mathcal C\cap K_R, p \neq q} \frac{1}{|p-q|^s}$$
 where $K_R$ is the cube of sidelength $R$ centered at the origin. When considering this quantity, there is however no implicit knowledge of the average density of points, contrarily to the situation of \cite{LebSer}. This is then easily extended to an energy for point processes $\bWsP$ by taking expectations.
 
 As in \cite{LebSer}, the starting point  of the LDP proof is a Sanov-type result that states that the logarithm of the  volume of configurations whose empirical processes lie in a small ball around a given tagged point process $\bP$ can be expressed as $\edb{(-N)}$ times an entropy  denoted $\ERS({\bP}|\Poisson)$. 
 \edb{As we shall show, $N^{-1-s/d}\HN(\XN)\approx \bWsP(\bP) +\bV(\bP)$ for a sufficiently large set of configurations $\XN$  near $\P$,
  where $\bV(\bP)$ is  a term corresponding to the external potential $V$.  
 Then this will suffice to obtain the LDP  since  the logarithm of the probability of the empirical field being close to $\bP$  is nearly $N$ times 
 \begin{equation*}
 - \beta  \edb{(\bWsP(\bP)+\bV(\bP))} - \ERS({\bP}|\Poisson),
 \end{equation*}
 up to an additive constant.
 The entropy can be expressed in terms of   $\bP^x$ (the process centered at $x$) as
\begin{equation} \ERS({\bP}|\Poisson)=\int (\ERS(\bP^x|\Poisson)-1)\, dx +1,
\end{equation}
where  $\ERS(P |\Poisson)$ is a ``specific relative entropy" with respect to the Poisson point process $\Poisson$. 
  Assuming that $\bP^x$  has an intensity $\rho(x)$, then the scaling properties of the energy $\bWsP$ (the fact that the energy scales like $\rho^{1+s/d}$ where $\rho$ is the density) and of the specific relative entropy $\ERS$ allow to transform this into 
 \begin{multline*}
 -  \int_{\Omega}\left(\beta \rho^{s/d} \WsP(\bP^x) \edb{+ \ERS(\sigma_{\rho(x)}(\bP^x)|\Poisson)} + \beta V(x) \right) \rho(x) \, dx \\ 
 \edb{-}  \int_{\Omega} \rho(x) \log \rho(x) \, dx,
 \end{multline*}}
 which is the desired rate function.
 Minimizing over $P$'s of intensity $\rho$ allows to obtain the rate function $I_\beta$ of \eqref{def:IbetaA}. 

 To run through this argument, we encounter the same difficulties as in \cite{LebSer}, i.e. the difficulty in replacing $\HN$ by $\bWsP$ due to the fact that $\HN$ is not continuous for the topology on empirical processes that we are considering. The lack of continuity of the interaction near the origin is dealt with by a truncation and regularization argument, similarly as in \cite{LebSer}. The lack of continuity due to the locality of the topology is handled thanks to the short-range nature of the Riesz interaction, by showing that  large microscopic boxes effectively do not interact, the ``self-screening" property alluded to before, via a shrinking procedure borrowed from \cite{HSAdv}.  We refer to Section \ref{sec4} for more detail.
 
 }

\section{General definitions}
All the hypercubes considered will have their sides parallel to some fixed choice of axes in $\R^d$. For $R > 0$ we let $\carr_R$ be the hypercube of center $0$ and sidelength $R$. If $A \subset \R^d$ is a Borel set we denote by $|A|$ its Lebesgue measure, and if $A$ is a finite set we denote by $|A|$ its cardinal. 

\subsection{(Random) (tagged) point configurations}
\label{sec:pointconfig}
\subsubsection{Point configurations}
We refer to \cite{dvj} for further details and proofs of the claims.
\begin{itemize}
\item  If $A \subset \R^d$, we denote by $\config(A)$ the set of locally finite point configurations in $A$ or equivalently the set of non-negative, purely atomic Radon measures on $A$ giving an integer mass to singletons. We abbreviate $\config(\R^d)$ as $\config$.

\item For $\C \in \config$, we will often write $\C$ for the Radon measure $\sum_{p \in \mathcal C} \delta_p$.

\item The sets $\config(A)$ are endowed with the topology induced by the weak convergence of Radon measures (also known as vague convergence). These topological spaces are Polish, and we fix a distance  $\dconfig$  on $\config$ which is compatible with the topology on $\config$ (and whose restriction on $\config(A)$ is also compatible with the topology on $\config(A)$).

\item For $x \in \R^d$ and $\C \in \config$ we denote by $\theta_x \cdot \C$ “the configuration $\C$ centered at $x$” (or “translated by $-x$”), namely
\begin{equation}\label{actiontrans}
\theta_x \cdot \C := \sum_{p \in \C} \delta_{p - x}.
\end{equation}
We will use the same notation for the action of $\Rd$ on Borel sets: if $A \subset \R^d$, we denote by $\theta_x \cdot A$ the translation of $A$ by the vector $-x$.

%\item For any integer $N$ we identify a configuration $\mc{C}$ that has $N$ points with all the $N$-tuples of points in $\Rd$ which correspond to $\mc{C}$ and if $A$ is a set of configurations with $N$ points we denote by $\Leb^{\otimes N}(A)$ the Lebesgue measure of the corresponding subset of $(\Rd)^N$.

\end{itemize}

\subsubsection{Tagged point configurations.}
\begin{itemize}
\item When $\Ac \subset \R^d$ is fixed, we define $\bconfig := \Ac \times \config$ as the set of “tagged” point configurations with tags in $\Ac$. 
\item We endow $\bconfig$ with the product topology and a compatible distance $\dbconfig$.
\end{itemize}

Tagged objects will usually be denoted with bars (e.g., $\bP$, $\overline{\mathbb{W}}$, \dots). 

\subsubsection{Random point configurations}
\begin{itemize}
\item We denote by $\pconfig$ the space of probability measures on $\config$; i.e., the set of laws of random point configurations.
\item The set $\pconfig$ is endowed with the topology of weak convergence of probability measures (with respect to the topology on $\config$), see \cite[Remark 2.7]{LebSer}.
\item We say that $P$ in $\pconfig$ is \textit{stationary} (and we write $P \in \psconfig$) if its law is invariant by the action of $\R^d$ on $\config$ as defined in \eqref{actiontrans}.
\end{itemize}

\subsubsection{Random tagged point configurations}
\label{sec:randomttaggedpoint}
\begin{itemize}
\item When $\Ac \subset \R^d$ is fixed, we define $\pbconfig$ as the space of measures $\bP$ on $\bconfig$ such that
\begin{enumerate}
\item The first marginal of $\bP$ is the Lebesgue measure on $\Ac$.
\item For almost every $x \in \Ac$, the disintegration measure $\bPx$ is an element of $\pconfig$.
\end{enumerate}
\item We say that $\bP$ in $\pbconfig$ is \textit{stationary} (and we write $\bP \in \psbconfig$) if $\bPx$ is in $\psconfig$ for almost every $x \in \Ac$.
\end{itemize}
Let us emphasize that, in general, the elements of $\pbconfig$ are \textit{not} probability measures on $\bconfig$ (e.g., the first marginal is the Lebesgue measure on $\Ac$).

\subsubsection{Density of a point configuration}
\label{sec:density}
\begin{itemize}
\item For $\C \in \config$, we define $\Dens(\C)$ (the \textit{density} of $\C$) as
\begin{equation}
\label{densDef}
\Dens(\C):=\liminf_{R\to\infty} \frac{\crd{\C \cap \carr_R}}{R^d}.
\end{equation}
\item For $m \in  [0, +\infty]$, we denote by  $\config_m$ the set of point configurations with density $m$. 
\item For $m \in (0, +\infty)$, the scaling map 
\begin{equation}
\label{def:sigmam} \sigma_m : \C \mapsto m\edb{^{1/d}} \C
\end{equation}
is a bijection of $\config_m$ onto $\config_1$, of inverse $\sigma_{1/m}$.
\end{itemize}

\subsubsection{Intensity of a random point configuration}
\begin{itemize} 
\item For $P \in \psconfig$, we define $\Intens(P)$ (the \textit{intensity} of $P$) as
\begin{equation*}
\Intens(P) := \Esp_{P} \left[ \Dens(\C) \right].
\end{equation*}

\item  We denote by $\psmconfig$ the set of laws of random point configurations $P \in \pconfig$ that are stationary and such that $\Intens(P) = m$. 
%The intensity exists (because of the stationarity assumption) and is such that
For $P\in \psmconfig$, the stationarity assumption implies the formula
\begin{equation*}
\Esp_P \left[ \int_{\R^d} \varphi\, d\C \right] = m \int_{\R^d} \varphi(x)\,  dx, \text{ for any $\varphi \in C^0_c(\R^d)$.}
\end{equation*}
%It is not hard to check that
%\begin{equation}\label{ExpDens}
%\Esp_P\left[ \Dens \right]=m,
%\end{equation}
%for any $P \in \psmconfig$.
\end{itemize}

\subsubsection{Intensity measure of a random tagged point configuration}
\label{sec:intensitymeasure}
\begin{itemize}
\item For $\bP$ in $\psbconfig$, we define $\bIntens(\bP)$ (the \textit{intensity measure} of $\bP$) as
\begin{equation*}
\bIntens(\bP)(x) = \Intens(\bPx),
\end{equation*}
which really should, in general, be understood in a dual sense: for any $f \in C_c(\R^d)$,
$$
\int f d\bIntens(\bP) := \int_{\Ac} f(x) \Intens(\bPx) dx.
$$
\item We denote by $\psunbconfig$ the set of laws of random tagged point configurations $\bP $ in  $\pbconfig$ which are stationary and such that 
\begin{equation*}
\int_{\Ac} \bIntens(\bP)(x)\, dx = 1
\end{equation*}
\item If $\bP$ has intensity measure $\rho$ we denote by $\bsigrho(\bP)$ the element of $\pbconfig$ satisfying
\begin{equation} \label{def:bsigrho}
\left(\bsigrho(\bP)\right)^x = \sigma_{\rho(x)}\left( \bP^x\right), \text{  for all $x \in \Ac$,}
\end{equation}
where $\sigma$ is as in \eqref{def:sigmam}.
\end{itemize}

\subsection{Specific relative entropy} \label{sec:ERS}
\begin{itemize}
\item Let $P$ be in $\psconfig$. The \textit{specific relative entropy} $\ERS[\Pst|\Poisson]$ of $\Pst$ with respect to $\Poisson$, the law of the Poisson point process of uniform intensity $1$, is given by
\begin{equation} \label{defERS}
\ERS[P|\Poisson] := \lim_{R \ti} \f{1}{|\carr_R|} \Ent\left(\Pst_{|\carr_R} | \Poisson_{|\carr_R} \right)
\end{equation}
where $ P_{|\carr_R}$ denotes the process induced on (the point configurations in) $\carr_R$, and   $\Ent( \cdot | \cdot)$ denotes the usual relative entropy (or Kullbak-Leibler divergence) of two probability measures defined on the same probability space.
\item  It is known (see e.g. \cite{MR3309619}) that
the  limit \eqref{defERS} exists as soon as $P$ is stationary, and also that the functional $P \mapsto \ERS[\Pst|\Poisson]$ is affine lower semi-continuous with compact sub-level sets (it is a good rate function).
\item Let us observe that the empty point process has specific relative entropy $1$ with respect to $\Poisson$.
\item If $P$ is in $\psmconfig$ we have (see \cite[Lemma 4.2.]{LebSer})
\begin{equation}\label{scalingent}
\ERS[P |\Poisson]=
\ERS [\sigma_m(P) |\Poisson]m + m \log m +1-m,
\end{equation}
where $\sigma_m(P)$ denotes the push-forward of $P$ by \eqref{def:sigmam}.
\end{itemize}

\subsection{Riesz energy of (random) (tagged) point configurations} \label{sec:energy}
\subsubsection{Riesz interaction}
We will use the notation $\Int$ (as “interaction”) in two slightly different ways:
\begin{itemize}
\item If $\C_1, \C_2$ are some fixed point configurations, we let $\Int[\C_1, \C_2]$ be the Riesz interaction between $\C_1$ and $\C_2$.
\begin{equation*}
\Int[\C_1, \C_2] := \sum_{p \in \C_1, \, q \in \C_2, p \neq q} \frac{1}{|p-q|^s}.
\end{equation*}
\item If $\C$ is a fixed point configuration and $A, B$ are two subsets of $\R^d$ we let $\Int[A, B](\C)$ to be the Riesz interaction between $\C \cap A$ and $\C \cap B$; i.e., 
\begin{equation*}
\Int[A,B](\C) := \Int[\C \cap A, \C \cap B] = \sum_{p \in \C \cap A, q \in \C \cap B, p \neq q} \frac{1}{|p-q|^s}.
\end{equation*}
\item Finally, if $\tau > 0$, we let $\Int_{\tau}$ be the truncation of the Riesz interaction  at distances less than $\tau$; i.e.,
\begin{equation} \label{def:Inttau}
\Int_{\tau}[\C_1, \C_2] := \sum_{p \in \C_1, q \in \C_2, |p-q| \geq \tau} \frac{1}{|p-q|^s}.
\end{equation}
\end{itemize}

\subsubsection{Riesz energy of a finite point configuration}
\label{sec:finitepointenergy}
\begin{itemize}
\item  Let $\om_N=(x_1,\ldots, x_N)$ be in $(\R^d)^N$. We define its Riesz $s$-energy as
\begin{equation} \label{def:Es}
 \Es(\om_N):=\Int[\om_N, \om_N] = \sum_{1 \leq i \neq j \leq N} \frac{1}{|x_i-x_j|^s}.
 \end{equation}
% Here and hereafter, we define an empty sum to be zero, so that $\Es(\om_N)=0$ if $N=0$ or $N=1$. 
 \item For $A \subset \R^d$, we consider the {\em $N$-point minimal  $s$-energy}
 \begin{equation} \label{def:EsAN}
 \Es(A, N) := \inf_{\om_N \in A^N} \Es(\om_N).
 \end{equation}
 \item The asymptotic minimal energy $\Csd$ is defined as
 \begin{equation}\label{def:Csd1}
\Csd :=\lim_{N\to \infty}\frac{\Es(\carr_1,N)}{N^{1+s/d}}.
\end{equation}
The limit in \eqref{def:Csd1} exists as a positive real number (see \cite{HSNotices,HSAdv}).

\item By scaling properties of the $s$-energy, it follows that
\begin{equation} \label{def:Csd2}
\lim_{N\to \infty}\frac{\Es(\carr_R,N)}{N^{1+s/d}}=\Csd R^{-s}.
\end{equation}
\end{itemize}
 
\subsubsection{Riesz energy of periodic point configurations}
We first extend the definition of the Riesz energy to the case of periodic point configurations.
\begin{itemize}
\item We say that $\Lambda \subset \R^d$ is a $d$-dimensional Bravais lattice if $\Lambda = U \Z^d$, for some nonsingular $d\times d$ real matrix $U$. A fundamental domain for $\Lambda$ is given by  $\Oml = U [-\hal, \hal)^d$, and the co-volume of $\Lambda$ is $|\Lambda| :=\text{vol}(\Oml) = |\det U|$. 

\item If $\C$ is a point configuration (finite or infinite) and $\Lambda$ a lattice, we denote by $\C + \Lambda$ the configuration $\{ p + \lambda \mid p \in \C, \lambda \in \Lambda \}$. We say that $\C$ is $\Lambda$-periodic if $\C + \Lambda = \C$.

\item If $\C$ is $\Lambda$-periodic, it is easy to see that we have $\C = \left(\C \cap \Oml\right) + \Lambda$. The density of $\C$ is thus given by
\begin{equation*}
\Dens(\C) = \frac{ |\C \cap \Oml|}{|\Lambda|}
\end{equation*}
\end{itemize}

Let $\Lambda$ be a lattice and $\om_N = \{x_1, \dots, x_N\} \subset \Oml$. 
\begin{itemize}
\item We define, as in \cite{HSS} for $s > d$, the {\em $\Lambda$-periodic $s$-energy of $\om_N$} as
\begin{equation}\label{perEndef}
\Esl (\om_N):=\sum_{x \in \om_N} \sum_{\substack{y\in \omega_N+\Lambda\\ y\neq x}}\frac{1}{|x-y|^s}.
\end{equation}
\item It follows (cf. \cite{HSS}) that $\Esl(\om_N)$ can be re-written as
\begin{equation} \label{Eslwithzeta}
\Esl(\om_N) = N\zeta_{\Lambda}(s)+\sum_{x\neq y\in \omega_N}\zeta_{\Lambda}(s,x-y),
\end{equation}
where 
$$\zeta_{\Lambda}(s)=\sum_{0\neq v\in \Lambda} {|v|^{-s}}$$
 denotes the {\em Epstein zeta  function} and 
$$\zeta_{\Lambda}(s,x):=\sum_{v\in \Lambda} {|x+v|^{-s}}$$
 denotes the {\em Epstein-Hurwitz zeta function} for the lattice $\Lambda$.
\item Denoting the {\em minimum $\Lambda$-periodic   $s$-energy} by
\begin{equation}\label{minperEndef}
\EslN :=\min_{\om_N \in \Oml^N} \Esl(\om_N),
\end{equation}
it is shown in \cite{HSS} that
\begin{equation}\label{perEnLim}
\lim_{N\to \infty}\frac{\EslN}{N^{1+s/d}}= \Csd |\Lambda|^{-s/d},
\end{equation}
where $\Csd$ is as in \eqref{def:Csd1}.\\
\end{itemize}

The constant $\Csd$ for $s>d$ appearing in \eqref{def:Csd1} and \eqref{perEnLim} is known only in the case $d=1$  where  $C_{s,1}=\zeta_{\Z}(s)=2\zeta(s)$ and $\zeta(s)$ denotes the classical Riemann zeta function.
  For dimensions $d=2, 4, 8$, and $24$,  it has been conjectured (cf. \cite{cohn2007universally, brauchart2012next} and references therein) that $C_{s,d}$ for $s>d$ is also given by an Epstein zeta function, specifically, that  $\Csd=\zeta_{\Lambda_d}(s)$ for $\Lambda_d$ denoting the  equilateral triangular (or hexagonal) lattice, the $D_4$ lattice, the $E_8$ lattice, and the Leech lattice (all scaled to have co-volume 1) in the dimensions $d=2, 4, 8,$ and 24, respectively.  
 
\subsubsection{Riesz energy of an infinite point configuration}
\begin{itemize}
\item Let  $\C$ in $\config$ be an (infinite) point configuration.  We define its Riesz $s$-energy as 
\begin{equation}
\label{def:Ws}
\Ws( \C) := \liminf_{R\to \infty} \frac{1}{R^d} \sum_{p \neq q \in \C \cap \carr_R} \frac{1}{|p-q|^s} = \liminf_{R \ti} \frac{1}{R^d} \Int[\carr_R, \carr_R](\C).
\end{equation}
If $\C =\emptyset$, we define $\Ws(\C)=0$.  The $s$-energy is non-negative and can be $+ \infty$.
\item We have, for any $\C$ in $\config$ and any $m \in (0, + \infty)$
\begin{equation}
\label{scalingw1}
\Ws(\sigma_m \C)= m^{-(1+s/d)} \Ws(\mc{C}).
\end{equation}
\end{itemize}

It is not difficult to verify (cf. \cite[Lemma 9.1]{cohn2007universally}), that if $\Lambda$ is a lattice and $\om_N$ is a $N$-tuple of points in $\Oml$ we have
\begin{equation}\label{WELambda}
\Ws (\om_N +\Lambda) = \frac{1}{|\Lambda|} \Esl (\om_N).
\end{equation}
In particular, we have (in view of \eqref{Eslwithzeta})
\begin{equation}\label{Wlambda}
\Ws(\Lambda)=|\Lambda|^{-1} \zeta_{\Lambda}(s).
\end{equation} 

\subsubsection{Riesz energy for laws of random point configurations}
\label{sec:energyerandompoint}
\begin{itemize}
\item Let $P$ be in $\pconfig$, we define its Riesz $s$-energy as 
\begin{equation}
\label{def:WsP1a}
\WsP (P):= \liminf_{R \ti} \frac{1}{R^d} \E_{P}\left[ \Int[\carr_R, \carr_R](\C) \right].
\end{equation}
\item Let $\bP$ be in $\pbconfig$, we define its Riesz $s$-energy as
\begin{equation}
\label{def:bWsP}
\bWsP (\bP) := \int_{\Ac} \WsP( \bPx)\, dx.
\end{equation}
\item Let $\bP$ in $\pbconfig$ with intensity measure $\rho$. It follows from \eqref{scalingw1}, \eqref{def:bWsP} and the definition \eqref{def:bsigrho} that
\begin{equation} \label{scalingw2}
\bWsP \left( \bP \right) = \int_{\Ac} \rho(x)^{1+s/d}\, \WsP\left( \left(\bsigrho(\bP)\right)^x \right) dx.
\end{equation}
\end{itemize}

Let us emphasize that we define $\WsP$ as in \eqref{def:WsP1a} and \textit{not} by $\Esp_{P}[\Ws]$. Fatou's lemma easily implies that 
\begin{equation} \label{Wsfatou}
\Esp_{P} [ \Ws] \leq \WsP(P)
\end{equation}
 and in fact, in the stationary case, we may show that equality holds (see Corollary \ref{coro:Fatouegal}). 

\subsubsection{Expression in terms of the two-point correlation function}
Let $P$ be in $\pconfig$ and let us assume that the two-point correlation function of $P$, denoted by $\rho_{2, P}$ exists in some distributional sense. We may easily express the Riesz energy of $P$ in terms of $\rho_{2,P}$ as follows
\begin{equation} \label{Wrho2P}
\WsP(P) = \liminf_{R \ti} \frac{1}{R^d}  \int_{\carr_R \times \carr_R}  \frac{1}{|x-y|^s} \rho_{2,P}(x,y) dx dy.
\end{equation} 
If $P$ is stationary, the expression can be simplified as
\begin{equation} \label{Wrho2P2}
\WsP(P) = \liminf_{R \ti} \int_{[-R, R]^d} \frac{1}{|v|^s} \rho_{2,P}(v) \prod_{i=1}^d \left( 1 - \frac{|v_i|}{R} \right) dv \, ,
\end{equation}
where $\rho_{2,P}(v) = \rho_{2,P}(0,v)$ (we abuse notation and see $\rho_{2,P}$ as a function of one variable, by stationarity) and $v = (v_1, \dots, v_d)$. 
Both \eqref{Wrho2P} and \eqref{Wrho2P2} follow from the definitions and easy manipulations, proofs (in a slightly different context) can be found in \cite{leble2016logarithmic}. \ed{Let us emphasize that the integral in the right-hand side of \eqref{Wrho2P} is on two variables, whereas the one in \eqref{Wrho2P2} is a single integral, obtained by using stationarity and applying Fubini's formula, which gives the weight $\prod_{i=1}^d \left( 1 - \frac{|v_i|}{R} \right)$.}

\subsection{The rate functions}
\label{sec:ratefunction}
\subsubsection{Definitions}
\begin{itemize}
\item For $\P$ be in $\pbconfig$, we define 
\begin{equation*}
\bV(\bP) := \int V(x) d\left(\bIntens(\bP)\right)(x).
\end{equation*}
This is the energy contribution of the potential $V$.
\item For $\bP$ be in $\psunbconfig$, we define
\begin{equation} \label{def:fbarbeta}
\fbarbeta(\bP) := \beta \left(  \bWsP(\bP) + \bV(\bP) \right) + \int_{\Ac} \left(\ERS[\bPx | \Poisson] -1\right) dx +1.
\end{equation}
It is a free energy functional, the sum of an energy term $ \bWsP(\bP) + \bV(\bP)$ weighted by the inverse temperature $\beta$ and an entropy term.

\item If $\rho$ is a probability density we define $\Ibeta(\rho)$ as
\ed{
\begin{multline} \label{def:IbetaA}
\Ibeta(\rho) :=  \int_{\Ac}  \inf_{P \in \edb{\mathcal{P}_{stat,\rho(x)}(\config) }} \left(\beta \WsP(P) +\ERS[P|\Poisson] -1\right)dx \\ + \beta \int_{\Ac} \rho(x) V(x)\, dx +1,\qquad
\end{multline}
\edb{which can be written} as
\begin{multline} \label{def:Ibeta}
 \Ibeta(\rho) =  \int_{\Ac} \rho(x) \inf_{P \in \psunconfig} \left( \beta \rho(x)^{s/d} \WsP(P) +\ERS[P|\Poisson] \right)dx \\ + \beta \int_{\Ac} \rho(x) V(x)\, dx + \int_{\Ac}  \rho(x) \log \rho(x) \, dx.
\end{multline}
\edb{This last equation may seem} more complicated but \edb{note that} the $\inf$ inside the integral is taken on a fixed set, independent of $\rho$.}
The rate function $\Ibeta$ is obtained in Section \ref{sec:contractionprinciple} as a \textit{contraction} (in the language of Large Deviations theory, see e.g. \cite[Section 3.1]{MR3309619}) of the functional $\fbarbeta$, \ed{and \eqref{def:Ibeta} follows from \eqref{def:IbetaA} by scaling properties of $\WsP$ and $\ERS[ \cdot | \Poisson]$.}
\end{itemize}

\subsubsection{Properties}
\begin{prop} \label{prop:ratefunction}
For all $\beta > 0$, the functionals  $\fbarbeta$ and $\Ibeta$ are good rate functions. Moreover, $\Ibeta$ is strictly convex.
\end{prop}
\begin{proof}
It is proven in Proposition \ref{prop:lsc} that $\bWsP$ is lower semi-continuous on $\psunbconfig$. As for $\bV$, we may observe that, if $\bP \in \psunbconfig$
\begin{equation*}
\bV(\bP) = \int_{\Ac \times \config} \left( V(x) |\C \cap \carr_1| \right) d\bP(x, \C),
\end{equation*}
and that $(x, \C) \mapsto V(x) |\C \cap \carr_1|$ is lower semi-continuous  on $\bconfig$, thus $\bV$ is lower semi-continuous on $\psunbconfig$, moreover, it is known that $\ERS[\cdot | \Poisson]$ is lower semi-continuous (see Section \ref{sec:ERS}). Thus $\fbarbeta$ is lower semi-continuous. Since $\bWsP$  and $\bV$ are bounded below, the sub-level sets of $\fbarbeta$ are included in those of $\ERS[\cdot | \Poisson]$, which are known to be compact (see again Section \ref{sec:ERS}). Thus $\fbarbeta$ is a good rate function.

The functional $\Ibeta$ is easily seen to be lower semi-continuous, and since $\Ws$, $\ERS$ and $V$ are bounded below, the sub-level sets of $\Ibeta$ are included into those of $\int_{\Ac} \rho \log \rho$ which are known to be compact, thus $\Ibeta$ is a good rate function. 

To prove that $\Ibeta$ is strictly convex in $\rho$, it is enough to prove that the first term in the right-hand side of \eqref{def:Ibeta} is convex (the second one is clearly affine, and the last one is well-known to be strictly convex). We may observe that the map
$$
\rho \mapsto \beta \rho^{1+s/d}  \WsP(P) + \rho \,\ERS[P|\Poisson] - \rho
$$
is convex for all $P$ (because $\WsP(P)$ is non-negative), and the infimum of a family of convex functions is also convex, thus
$$
\rho \mapsto \inf_{P \in \psunconfig} \left( \beta \rho^{1+s/d} \WsP(P) + \rho \, \ERS[P|\Poisson] \right)
$$
is convex in $\rho$, which concludes the proof.
\end{proof}

\section{Preliminaries on the energy}
\subsection{General properties}

\subsubsection{Minimal energy of infinite point configurations}
In this section, we connect the minimization of $\Ws$ (defined at the level of infinite point configurations) with the asymptotics of the $N$-point minimal energy as presented in Section \ref{sec:finitepointenergy}. Let us recall that the class $\config_m$ of point configurations with mean density $m$ has been defined in Section \ref{sec:density}.
\begin{prop} \label{prop:WS}
We have
\begin{equation} \label{minA1}
\inf_{\C \in \config_1} \Ws(\C) = \min_{\C \in \config_1} \Ws(\C) = \Csd,
\end{equation}
where $\Csd$ is as in \eqref{def:Csd1}. Moreover, for any $d$-dimensional Bravais lattice $\Lambda$ of co-volume $1$, there exists a minimizing sequence $\{C_N\}_N$ for $\Ws$ over $\config_1$ such that $\C_N$ is $N^{1/d}\Lambda$-periodic for $N \geq 1$.
\end{prop}
\begin{proof}
Let $\Lambda$ be a $d$-dimensional Bravais lattice $\Lambda$ of co-volume $1$, and for any $N$ let $\om_N$ be a $N$-point configuration minimizing $E_{s,\Lambda}$. We define
\begin{equation*}
\C_N := N^{1/d} \left( \om_N + \Lambda \right).
\end{equation*}
By construction, $\C_N$ is a $N^{1/d} \Lambda$-periodic point configuration of density $1$. Using the scaling property \eqref{scalingw1} and \eqref{WELambda}, we have
\begin{equation*}
 \Ws (\C_N)= \frac{\Ws\left( \om_N+\Lambda\right)}{N^{1+s/d}} = \frac{\Esl(\om_N)}{N^{1+s/d}}.
\end{equation*}
On the other hand, we have by assumption $\Esl(\om_N) = \EslN$. Taking the limit $N \ti$ yields, in light of  \eqref{perEnLim}, $\lim_{N \ti} \Ws(\C_N) = \Csd$. In particular we have
\begin{equation} \label{infWsa}
\inf_{\C \in \config_1} \Ws(\C)  \leq \Csd.
\end{equation}
To prove the converse inequality, let us consider $\C$ in $\config_1$ arbitrary. We have by definition (see \eqref{def:Es} and \eqref{def:Ws}) and the scaling properties of $\Es$,
\begin{equation*}
\Ws(\C) =\liminf_{R \to \infty} \frac{\Es \left( \C \cap \carr_R \right)}{R^d} = \liminf_{R \to \infty}\frac{1}{R^{d+s}} \Es \left(\frac{1}{R}\C \cap \carr_1\right),
\end{equation*}
and, again by definition (see \eqref{def:EsAN})
\begin{equation*}
\Es \left(\frac{1}{R}\C \cap \carr_1\right) \geq \mathcal{E}_s\left(\carr_1,  |\C \cap \carr_R| \right).
\end{equation*}
We thus obtain
\begin{equation*}
\Ws(\C) \geq \liminf_{R \to \infty} \frac{\mathcal{E}_s\left(\carr_1, |\mc{C}\cap \carr_R| \right)}{\crd{\mc{C}\cap \carr_R}^{1+s/d}} \left(\frac{\crd{\mc{C}\cap \carr_R}}{R^d}\right)^{1+s/d}.
\end{equation*}
Using the definition \eqref{def:Csd1} of $\Csd$ we have
$$
\liminf_{R \ti}  \frac{\mathcal{E}_s\left(\carr_1, |\mc{C}\cap \carr_R| \right)}{\crd{\mc{C}\cap \carr_R}^{1+s/d}} \geq \Csd, 
$$
and by definition of the density, since $\C$ is in $\config_1$ we have
$$
\liminf_{R \ti} \left(\frac{\crd{\mc{C}\cap \carr_R}}{R^d}\right)^{1+s/d} = 1.
$$
It yields $\Ws(\C) \geq \Csd$ and so (in view of \eqref{infWsa})
 \begin{equation} \label{infWsb}
 \inf_{\mc{C}\in \mc{A}_1} \Ws(\mc{C})= \Csd.
 \end{equation}

It remains to prove that the infimum is achieved. Let us start with a sequence $\{\om_M\}_{M \geq 1}$ such that $\om_M$ is a $M^d$-point configuration in $K_M$ satisfying
\begin{equation} \label{omegaMminimise}
\lim_{M \ti} \frac{\Es(\omega_M)}{M^d} = \Csd.
\end{equation}
Such a sequence of point configurations exists by definition of $\Csd$ as in \eqref{def:Csd1}, and by the scaling properties of $\Es$. We define a configuration $\C$ inductively as follows. 
\begin{itemize}
\item Let $r_1, c_1, s_1 = 1$ and let us set $\C \cap \carr_{r_1}$ to be $\omega_1$.  
\item Assume that $r_N, s_N, c_N$ and $\C \cap \carr_{r_N}$ have been defined. We let 
\begin{equation} \label{choiceSN}
s_{N+1} = \lceil c_{N+1}r_N + (c_{N+1} r_N)^{\hal} \rceil,
\end{equation}
with $c_{N+1} > 1$ to be chosen later. We also let $r_{N+1}$ be a multiple of $s_{N+1}$ large enough, to be chosen later. We tile $\carr_{r_{N+1}}$ by hypercubes of sidelength $s_{N+1}$ and we define $\C \cap \carr_{r_{N+1}}$ as follows:
\begin{itemize}
\item In the central hypercube of sidelength $s_{N+1}$, we already have the points of $\C \cap \carr_{r_N}$ (because $r_N \leq s_{N+1}$) and we do not add any points. In particular, this ensures that each step of our construction is compatible with the previous ones.
\item In all the other hypercubes, we paste a copy of $\omega_{c_{N+1} r_N}$ “centered” in the hypercube in such a way that 
\begin{equation} \label{hypdistanhyper}
\text{all the points are at distance $\geq (c_{N+1} r_N)^{\hal}$ of the boundary}.
\end{equation}
This is always possible because $\omega_{c_{N+1} r_N}$ lives, by definition, in an hypercube of sidelength $c_{N+1} r_N$ and because we have chosen $s_{N+1}$ as in \eqref{choiceSN}. 
\end{itemize}
We claim that the number of points in $\carr_{r_{N+1}}$ is always less than $r_{N+1}^d$ (as can easily be checked by induction) and is bounded below by
\begin{equation*}
\left( \left(\frac{r_{N+1}}{s_{N+1}}\right)^d  -1 \right) (c_{N+1} r_N)^d.
\end{equation*}
Thus it is easy to see that if $c_{N+1}$ is chosen large enough and if $r_{N+1}$ is a large enough multiple of $s_{N+1}$, then 
\begin{equation} \label{nombrepointsSNetc}
\text{ the number of points in $r_{N+1}$ is $r_{N+1}^d (1 - o_N(1))$.}
\end{equation}
Let us now give an upper bound on the interaction energy $\Int[\carr_{r_{N+1}}, \carr_{r_{N+1}}](\C)$. We recall that we have tiled $\carr_{r_{N+1}}$ by hypercubes of sidelength $s_{N+1}$. 
\begin{itemize}
\item Each hypercube as a self-interaction energy given by $\Es(\omega_{c_{N+1} r_N})$, except the central one, whose self-interaction energy is bounded by $O(r_N^d)$ (as can be seen by induction).
\item The interaction of a given hypercube with the union of all the others can be controlled because, by construction (see \eqref{hypdistanhyper}) the configurations pasted in two disjoint hypercubes are far way from each other. We can compare it to
\begin{equation*}
 \int_{r =(c_{N+1} r_N)^{\hal}}^{+ \infty} \frac{1}{r^s} s_{N+1}^{d} r^{d-1} dr,
\end{equation*}
and an elementary computation shows that it is negligible with respect to $s_{N+1}^d$ (because $d < s$).
\end{itemize}

We thus have
\begin{equation*}
\Int[\carr_{r_{N+1}}, \carr_{r_{N+1}}](\C) \leq \left( \left(\frac{r_{N+1}}{s_{N+1}}\right)^d -1 \right) \Es(\omega_{c_{N+1} r_N}) + O(r_N^d) + \left(\frac{r_{N+1}}{s_{N+1}}\right)^d o_N\left(s_{N+1}^d \right).
\end{equation*}
We may now use  \eqref{omegaMminimise} and get that 
\begin{equation} \label{energieminimisante}
\frac{1}{r_{N+1}^d} \Int[\carr_{r_{N+1}}, \carr_{r_{N+1}}](\C) \leq \Csd + o_N(1).
\end{equation}
\end{itemize}
Let $\C$ be the point configuration constructed as above. Taking the limit as $N \ti$ in \eqref{nombrepointsSNetc} shows that $\C$ is in $\config_1$, and \eqref{energieminimisante}  implies that $\Ws(\C) \leq \Csd$, which concludes the proof of \eqref{minA1}.
\end{proof}

\subsubsection{Energy of random point configurations}
In the following lemma, we prove that for stationary $P$ the $\liminf$ defining $\WsP(P)$ as in \eqref{def:WsP1a} is actually a limit, and that the convergence is uniform of sublevel sets of $\WsP$ (which will be useful for proving lower semi-continuity).
\begin{lem}\label{lem:WsP}
Let $P$ be in $\psconfig$. The following limit exists in $[0, +\infty]$
 \begin{equation} \label{WsP2}
 \WsP(\P) := \lim_{R \ti} \frac{1}{R^d} \Esp_{P} \left[ \Int[\carr_R, \carr_R]\right].
 \end{equation}
Moreover we have as $R \ti$
\begin{equation} \label{erreurWsP2}
\left| \WsP(\P) - \frac{1}{R^d} \Esp_{P} \left[ \Int[\carr_R, \carr_R] \right] \right| \leq C \left(\WsP(P)^{\frac{2}{1+s/d}}  + \WsP(P) \right) o_R(1),
\end{equation}
with $o_R(1)$ depending only on $s,d$.
 \end{lem}
 \begin{proof}
We begin by showing that the quantity 
$$\frac{1}{n^d}  \Esp_{P} \left[ \Int[\carr_{n}, \carr_n](\C)\right]$$ 
is non-decreasing for integer values of $n$. 

For $n \geq 1$, let $\{\tK_v\}_{v \in \Z^d \cap \carr_n }$ be a tiling of $\carr_n$ by unit hypercubes, indexed by the centers $v \in \Z^d \cap \carr_n$ of the hypercubes, and let us split $\Int[\carr_n, \carr_n]$ as
\begin{equation*}
\Int[\carr_{n}, \carr_n] = \sum_{v , v'\in \Z^d \cap \carr_n} \Int[\tK_{v}, \tK_{v'}].
\end{equation*}
Using the stationarity assumption and writing $v = (v_1, \dots, v_d)$ and   $|v|:=\max_i |v_i|$, we obtain
\begin{equation*}
\Esp_{P} \left[ \sum_{v, v' \in \Z^d \cap \carr_n} \Int[\tK_{v}, \tK_{v'}] \right] = \sum_{v \in \Z^d \cap \carr_{2n}} \Esp_{P} \left[\Int[\tK_{0}, \tK_{v}] \right] \prod_{i=1}^d (n - |v_i|).
\end{equation*}
 We  thus get
\begin{equation} \label{WsPpremier}
\frac{1}{n^d} \Esp_{P} \left[ \Int[\carr_{n}, \carr_n] \right] = \sum_{v \in \Z^d \cap \carr_{2n}} \Esp_{P} \left[\Int[\tK_{0}, \tK_{v}] \right]  \prod_{i=1}^d \left(1 - \frac{|v_i|}{n}\right), 
\end{equation}
and it is clear that this quantity is non-decreasing in $n$, in particular the limit as $n \ti$ exists in $[0, + \infty]$. We may also observe that $R \mapsto \Int[\carr_R, \carr_R]$ is non-decreasing in $R$. It is then easy to conclude that the limit of \eqref{WsP2} exists in $[0, + \infty]$.

Let us now quantify the speed of convergence. First, we observe that for $|v| \geq 2$ we have
\begin{equation*}
\Esp_{P} \left[\Int[\tK_{0}, \tK_{v}] \right] \leq O\left(\frac{1}{|v-1|^s}\right) \Esp_{P}[ N_0 N_v],
\end{equation*}
where $N_0, N_v$ denotes the number of points in $\tK_0, \tK_v$. Indeed, the points of $\tK_{0}$ and $\tK_{v}$ are at distance at least $|v-1|$ from each other (up to a multiplicative constant depending only on $d$).

On the other hand, Hölder's inequality and the stationarity of $P$ imply
\begin{equation*}
\|N_0 N_v\|_{L^1(P)} \leq \|N_0 \|_{L^{1+s/d}(P)} \|N_v \|_{L^{1+s/d}(P)} =  \|N_0 \|^2_{L^{1+s/d}(P)},
\end{equation*}
and thus we have $\Esp_{P}[ N_0 N_v] \leq \Esp_{P}[N_0]^{\frac{2}{1+s/d}}$. On the other hand, it is easy to check that for $P$ stationary, 
$$\Esp_{P} [N_0^{1+s/d}] \leq C \WsP(P)$$ 
for some constant $C$ depending on $d,s$. Indeed, the interaction energy in the hypercube $\tK_0$ is bounded below by some constant times $N_0^{1+s/d}$, and \eqref{WsPpremier} shows that
$$
\WsP(P) \geq \Esp_{P} \left[\Int[\tK_{0}, \tK_0]\right].
$$

We thus get
\begin{multline*}
\WsP(P) - \sum_{v \in \Z^d \cap \carr_{2n}} \Esp_{P} \left[\Int[\tK_{0}, \tK_{v}] \right]  \prod_{i=1}^d \left(1 - \frac{|v_i|}{n}\right) \\ \leq  \WsP(P)^{\frac{2}{1+s/d}} \left(\sum_{v \in \Z^d \cap \carr_{2n}, |v| \geq 2} \frac{1}{|v-1|^s} \left(1 - \prod_{i=1}^d \left(1 - \frac{|v_i|}{n}\right) \right) + \sum_{|v| \geq 2n} \frac{1}{|v|^s}\right) \\
+ \frac{1}{n} \sum_{|v| =1}  \Esp_{P} \left[\Int[\tK_{0}, \tK_{v}] \right] .
\end{multline*}
It is not hard to see that the parenthesis in the right-hand side goes to zero as $n \ti$. On the other hand, we have
$$
\sum_{|v| =1}  \Esp_{P} \left[\Int[\tK_{0}, \tK_{v}] \right]  \leq \WsP(P).
$$
 Thus we obtain
\begin{equation*}
\WsP(P)  - \frac{1}{n^d} \Esp_{P} \left[ \Int[\carr_{n}, \carr_n] \right] \leq \left(\WsP(P)^{\frac{2}{1+s/d}}  +  \WsP(P)\right) o_n(1),
\end{equation*}
with a $o_n(1)$ depending only on $d,s$ and it is then not hard to get \eqref{erreurWsP2}.
\end{proof}

For any $R > 0$, the quantity $\Int[\carr_{R}, \carr_R]$ is continuous and bounded below on $\config$, thus the map $$P \mapsto \frac{1}{R^d} \Esp_{P} \left[ \Int[\carr_{R}, \carr_R] \right]$$ is lower semi-continuous on $\pconfig$. The second part of Lemma \ref{lem:WsP} shows that we may approximate $\WsP(P)$ by $\frac{1}{R^d} \Esp_{P} \left[ \Int[\carr_{R}, \carr_R] \right]$ up to an error which $o_R(1)$, uniformly on sub-level sets of $\WsP$. The next proposition follows easily.
\begin{prop} \label{prop:lsc}
\begin{enumerate}
\item The functional $\WsP$ is lower semi-continuous on $\psunconfig$.
\item The functional $\bWsP$ is lower semi-continuous on $\psunbconfig$.
\end{enumerate}
\end{prop}
 
 We may also prove the following equality (which settles a question raised in Section \ref{sec:energyerandompoint}).
 \begin{coro} \label{coro:Fatouegal}
 Let $P$ be in  $\psunconfig$, then we have
 $$
 \WsP(P) = \lim_{R \ti} \frac{1}{R^d} \Esp_{P} \left[ \Int[\carr_R, \carr_R](\C) \right] = \Esp_{P} \left[ \liminf_{R \ti} \frac{1}{R^d} \Int[\carr_R, \carr_R](\C) \right].
 $$
 \end{coro}
 \begin{proof}
 As was observed in \eqref{Wsfatou}, Fatou's lemma implies that
 $$
 \Esp_{P} \left[ \liminf_{R \ti} \frac{1}{R^d} \Int[\carr_R, \carr_R](\C) \right] \leq \lim_{R \ti} \frac{1}{R^d} \Esp_{P} \left[ \Int[\carr_R, \carr_R](\C) \right] = \WsP(P),
 $$
 (the last equality is by definition). On the other hand, with the notation of the proof of Lemma \ref{lem:WsP}, we have for any integer $n$ and any $\C$ in $\config$
 $$
 \frac{1}{n^d} \Int[\carr_n, \carr_n](\C) = \frac{1}{n^d} \sum_{v, v' \in \Z^d \cap \carr_R} \Int[\tK_{v}, \tK_{v'}], 
 $$
 and the right-hand side is dominated under $P$ (as observed in the previous proof), thus the dominated convergence theorem applies.  
 \end{proof}
 
\subsection{Derivation of the infinite-volume limit of the energy}
The following result is central in our analysis. It connects the asymptotics of the $N$-point interaction energy $\{\HN(\XN)\}_N$ with the infinite-volume energy $\bWsP(\bP)$ of an infinite-volume object: the limit point $\bP$ of the tagged empirical processes $\{\bEmp_N(\XN)\}_N$.   

\begin{prop} \label{prop:glinf} For any $N \geq 1$, let $\XN = (x_1, \dots, x_N)$ be in $\Ac^N$, let $\mu_N$ be the empirical measure and $\bP_N$ be the tagged empirical process associated to $\XN$; i.e.,
\begin{equation*}
\mu_N := \emp(\XN), \quad \bP_N := \bEmp_N(\XN),
\end{equation*}
as defined in \eqref{def:emp} and \eqref{def:Emp}. Let us assume that 
\begin{equation*} 
\liminf_{N \ti} \frac{\HN(\XN)}{N^{1+s/d}} < + \infty.
\end{equation*}
Then, up to extraction of a subsequence,
\begin{itemize}
\item $\{\mu_N\}_N$ converges weakly to some $\mu$ in $\probas(\Ac)$,
\item $\{\bP_N\}_N$ converges weakly to some $\bP$ in $\psunbconfig$,
\item $\Intens(\bP) = \mu$.
\end{itemize}   Moreover we have
 \begin{equation} \label{conc:glinf}
\liminf_{N\to \infty} \frac{ \HN(x_1, \dots, x_N)}{N^{1+s/d}} \ge  \bWsP(\bP)+ \bV(\bP).
\end{equation}
\end{prop}
\begin{proof}
Up to extracting a subsequence, we may assume that $\HN(\XN)= O(N^{1+s/d})$. 
First, by positivity of the Riesz interaction, we have for $N \geq 1$
\begin{equation*}
\int_{\Ac} V \, d\mu_N \leq \frac{\HN(\XN)}{N^{1+s/d}},
\end{equation*}
and thus $\int_{\Ac} V \, d\mu_N$ is bounded. By \eqref{ass:regV} and \eqref{ass:croissanceV} we know that $V$ is bounded below and has compact sub-level sets. An easy application of Markov's inequality shows that $\{\mu_N\}_N$ is tight, and thus it converges (up to another extraction). It is not hard to check that $\{\bP_N\}_N$ converges (up to extraction) to some $\bP$ in $\pbconfig$ (indeed the average number of points per unit volume is constant, which implies tightness, see e.g. \cite[Lemma 4.1]{LebSer}) whose stationarity is clear (see again e.g. \cite{LebSer}).

Let $\brho$ be the intensity measure of $\bP$ (in the sense of Section \ref{sec:intensitymeasure}), we want to prove that $\brho = \mu$ (which will in particular imply that $\bP$ is in $\psunbconfig$). It is a general fact that $\brho \leq \mu$ (see e.g. \cite[Lemma 3.7]{leble2015large}), but it could happen that a positive fraction of the points cluster together, resulting in the existence of a singular part in $\mu$ which is missed by $\brho$ so that $\brho < \mu$. However, in the present case, we can easily bound the moment (under $\bP_N$) of order $1 + s/d$ of the number of points in a given hypercube $\carr_R$. Indeed, let $\{\tilde{K}_i\}_{i \in I}$ be a covering of $\Ac$ by disjoint  hypercubes of sidelength $RN^{-1/d}$, and let $n_i = N\mu_N\left(\tilde{K}_i\right)$ denote the number of points from $\XN$ in $\tilde{K}_i$. We have, by positivity of the Riesz interaction
\begin{equation*}
\HN(\XN) \geq \sum_{i \in I} \Int[\tilde{K}_i, \tilde{K}_i] \geq C\sum_{i \in I}  \frac{n_i^{1+s/d}N^{s/d}}{R^s},
\end{equation*}
for some constant $C>0$ (depending only on $s$ and $d$) because the minimal interaction energy of $n$ points in $\tilde{K}_i$ is proportional to $\frac{n^{1+s/d}N^{s/d}}{R^s}$ (see \eqref{def:Csd1}, \eqref{def:Csd2}). Since $\HN(\XN) = O(N^{1+s/d})$ by assumption, we get that $\sum_{i \in I} n_i^{1+s/d} = O(N)$, with an implicit constant depending only on $R$. It implies that $x \mapsto N\mu_N \left( B(x, RN^{-1/d}) \right)$ is uniformly (in $N$) locally integrable on $\Ac$ for all $R > 0$, and arguing as in \cite[Lemma 3.7]{leble2015large} we deduce that $\brho = \mu$.

We now turn to proving \eqref{conc:glinf}. Using the positivity and scaling properties of the Riesz interaction and a Fubini-type argument we may write, for any $R > 0$
\begin{equation*}
\Int[\Ac, \Ac](\XN) \geq N^{1+s/d} \int_{\Ac \times \config} \frac{1}{R^d} \Int[\carr_R, \carr_R](\C) d\bP_N(x, \C).
\end{equation*}

Of course we have, for any $M > 0$,
\begin{equation*}
\int_{\Ac \times \config} \frac{1}{R^d} \Int[\carr_R, \carr_R](\C)d\bP_N(x, \C) \geq \int_{\Ac \times \config} \frac{1}{R^d} \left( \Int[\carr_R, \carr_R](\C) \wedge M\right) d\bP_N(x, \C),
\end{equation*}
and thus the weak convergence of $\bP_N$ to $\bP$ ensures that
\begin{equation*}
\int_{\Ac \times \config} \frac{1}{R^d} \Int[\carr_R, \carr_R](\C) d\bP_N(x, \C) \geq \int_{\Ac \times \config} \frac{1}{R^d} \left(\Int[\carr_R, \carr_R](\C) \wedge M\right) d\bP(x, \C) + o_N(1).
\end{equation*}
Since this is true for all $M$ we obtain
\begin{equation*}
\liminf_{N \ti} \frac{\Int[\Ac, \Ac](\XN)}{N^{1+s/d}} \geq \int_{\Ac \times \config} \frac{1}{R^d} \left( \Int[\carr_R, \carr_R](\C) \right) d\bP(x, \C).
\end{equation*}
Sending $R$ to $+ \infty$ and using Proposition \ref{prop:WS} we get
\begin{equation} \label{limiteint}
\liminf_{N \ti} \frac{\Int[\Ac, \Ac](\XN)}{N^{1+s/d}} \geq \liminf_{R \ti} \int_{\Ac \times \config} \frac{1}{R^d} \left( \Int[\carr_R, \carr_R](\C) \right) d\bP(x, \C) =: \bWsP(\bP).
\end{equation}
On the other hand, the weak convergence of $\mu_N$ to $\mu$ and Assumption \ref{ass:regV} ensure that
\begin{equation} \label{limiteV}
\liminf_{N \ti} \int_{\Ac} V\, d\mu_N \geq \int_{\Ac} V\, d\mu.
\end{equation}
Combining \eqref{limiteint} and \eqref{limiteV} gives \eqref{conc:glinf}.
\end{proof}

Proposition \ref{prop:glinf} can be viewed as a $\Gamma$-$\liminf$ result (in the language of $\Gamma$-convergence). We will prove later (e.g. in Proposition \ref{quasicontinu}, which is in fact a much stronger statement) the corresponding $\Gamma$-$\limsup$.

\section{Proof of the large deviation principles}\label{sec4}
As in \cite{LebSer}, the main obstacle for proving Theorem \ref{theo:LDPemp} is to deal with the lack of upper semi-continuity of the interaction, namely that there is no upper bound of the type 
\begin{equation*}
\HN(\XN) \lesssim N^{1+s/d} \left( \bWsP(\bP) + \bV(\bP) \right)
\end{equation*}
which holds in general under the mere condition that $\bEmp_N(\XN) \approx \bP$ (cf. \eqref{def:Emp} for a definition of the tagged empirical process). This yields a problem for proving the large deviations lower bound (in contrast, \textit{lower} semi-continuity holds and the proof of the large deviations upper bound is quite simple). Let us briefly explain why.

Firstly, due its singularity at $0$, the interaction is not uniformly continuous with respect to the topology on the configurations. Indeed a pair of points at distance $\epsilon$ yields a $\epsilon^{-s}$ energy but a pair of points at distance $2 \epsilon$ has energy $(2\epsilon)^{-s}$, with $|\epsilon^{-s} - (2\epsilon)^{-s} | \to \infty$, although these two point configurations are very close for the topology on $\config$.

Secondly, the energy is non-additive: we have in general
\begin{equation*}
\Int[\C_1 \cup \C_2, \C_1 \cup \C_2] \neq \Int[\C_1, \C_1] + \Int[\C_2, \C_2]. 
\end{equation*}
Yet the knowledge of $\bEmp_N$ (through the fact that $\bEmp_N(\XN) \in B(\bP, \epsilon)$) yields only   \textit{local} information on  $\XN$, and does not allow one to reconstruct $\XN$ \textit{globally}. Roughly speaking, it is like partitioning $\Ac$ into hypercubes and having a family of point configurations, each belonging to some hypercube, but without knowing the precise configuration-hypercube pairing. Since the energy is non-additive (there are non trivial hypercube-hypercube interactions in addition to the hypercubes' self-interactions), we cannot (in general) deduce $\HN(\XN)$ from the mere knowledge of the tagged empirical process.

In Section \ref{sec:LDPLB}, the singularity problem is dealt with by using a regularization procedure similar to that of \cite{LebSer}, while the non-additivity is shown to be negligible due to the short-range nature of the Riesz potential for $s > d$. 

\subsection{A LDP for the reference measure}
Let $\Leba$ be the Lebesgue measure on $\Ac^N$, and let $\bQN$ be the push-forward of $\Leba$ by the “tagged empirical process” map $\bEmp_N$ defined in \eqref{def:Emp}. Let us recall that $\Ac$ is not necessarily bounded, hence $\Leba$ may have an infinite mass and thus there is no natural way of making $\bQN$ a probability measure.

\begin{prop} \label{prop:Sanovreference}
Let $\bP$ be in $\psunbconfig$. We have
\begin{multline} \label{SanovReference}
\lim_{\epsilon \t0} \liminf_{N \ti} \frac{1}{N} \log \bQN\left( B(\bP, \epsilon) \right) = \lim_{\epsilon \t0} \limsup_{N \ti} \frac{1}{N} \log \bQN\left( B(\bP, \epsilon) \right) 
\\ = - \int_{\Ac} \left(\ERS[\bPx | \Poisson] -1\right) dx - 1.
\end{multline}
\end{prop}
\ed{We recall that $\bPx$ is the disintegration measure of $\bP$ at the point $x$, or the “fiber at $x$" (which is a measure on $\config$) of $\bP$ (which is a measure on $\Omega \times \config$), see Section \ref{sec:randomttaggedpoint}.}

\begin{proof} 
If $\Ac$ is bounded, Proposition \ref{prop:Sanovreference} follows from the analysis of \cite[Section 7.2]{LebSer}, see in particular \cite[Lemma 7.8]{LebSer}. The only difference is that the Lebesgue measure on $\Ac$ used in \cite{LebSer} is normalized, which yields an additional factor of $\log |\Ac|$ in the rate function. The proof extends readily to a non-bounded $\Ac$ because the topology of weak convergence on $\pbconfig$ is defined with respect to test functions which are compactly supported on $\Ac$. 
\end{proof}

\subsection{A LDP upper bound}
\begin{prop} \label{prop:LDPUB}
Let $\bP$ be in $\psunbconfig$.
We have
\begin{equation} \label{LDPUB}
\lim_{\epsilon \t0} \limsup_{N \ti} \frac{1}{N} \log \fPNbeta ( B(\bP, \epsilon)) \leq - \fbarbeta(\bP) + \limsup_{N \ti} \left(- \frac{\log \ZNbeta}{N}\right) .
\end{equation}
\end{prop}
\begin{proof}
Using the definition of $\fPNbeta$ as the push-forward of $\PNbeta$ by $\bEmp_N$ we may write
\begin{equation*}
\fPNbeta ( B(\bP, \epsilon)) = \frac{1}{\ZNbeta} \int_{\Ac^N \cap \{\bEmp_N(\XN) \in B(\bP, \epsilon)\}} \exp\left(-\beta N^{-s/d} \HN(\XN)\right) d\XN.
\end{equation*}
From Proposition \ref{prop:glinf} and Proposition \ref{prop:lsc} we know that for any sequence $\XN$ such that $\bEmp_N(\XN) \in B(\bP, \epsilon)$ we have
\begin{equation*}
\liminf_{N \ti} \frac{\HN(\XN)}{N^{1+s/d}} \geq \bWsP(\bP) + \bV(\bP) + o_{\epsilon}(1).
\end{equation*}
We may thus write
\begin{multline*}
\limsup_{N \ti} \frac{1}{N} \log \fPNbeta ( B(\bP, \epsilon)) \leq   - \beta \left(  \bWsP(\bP) + \bV(\bP) \right)   \\
+  \limsup_{N \ti}  \int_{\Ac^N \cap  \{\bEmp_N(\XN) \in B(\bP, \epsilon)\}} d\XN + \limsup_{N \ti}  \left(  - \frac{\log \ZNbeta}{N}  \right) + o_{\epsilon}(1).
\end{multline*}
Using Proposition \ref{prop:Sanovreference} we know that
\begin{equation*}
\limsup_{N \ti} \frac{1}{N} \log \int_{\Ac^N \cap  \{\bEmp_N(\XN) \in B(\bP, \epsilon)\}} d\XN =  - \int_{\Ac} \left(\ERS[\bPx | \Poisson] -1\right) - 1 + o_{\epsilon}(1).
\end{equation*}
We thus obtain, sending $\epsilon \t0$
\begin{multline*}
\limsup_{N \ti} \frac{1}{N} \log \fPNbeta ( B(\bP, \epsilon)) \leq   - \beta \left(  \bWsP(\bP) + \bV(\bP) \right) - \int_{\Ac} \left(\ERS[\bPx | \Poisson] -1\right) - 1 \\
+ \limsup_{N \ti} \left(- \frac{\log \ZNbeta}{N}\right),
\end{multline*}
which, in view of the definition of $\fbarbeta$ as in \eqref{def:fbarbeta}, yields \eqref{LDPUB}.
\end{proof}

\subsection{A LDP lower bound} \label{sec:LDPLB}
The goal of the present section is to prove a matching LDP lower bound:
\begin{prop} \label{prop:LDPLB}
Let $\bP$ be in $\psunbconfig$.
We have
\begin{equation} \label{LDPLB}
\lim_{\epsilon \t0} \liminf_{N \ti} \frac{1}{N} \log \fPNbeta ( B(\bP, \epsilon)) \geq - \fbarbeta(\bP) + \liminf_{N \ti} \left(- \frac{\log \ZNbeta}{N}\right).
\end{equation}
\end{prop}

For $N \geq 1$ and $\delta > 0$, let us define the set $T_{N, \delta}(\bP)$ as
\begin{equation} \label{def:TNdelta}
T_{N, \delta}(\bP) = \left\lbrace \XN   \mid   \frac{\HN(\XN)}{N^{1+s/d}} \leq \fbarbeta(\bP) + \delta \right\rbrace.
\end{equation}
We will rely on the following result:
\begin{prop}
\label{prop:quasicontinu}
Let $\bP$ be in $\psunbconfig$. For all $\epsilon, \delta >0$  we have
\begin{equation}
\label{quasicontinu}
\begin{split}
\liminf_{N \ti} \frac{1}{N} \log \Leba &\left( \bEmp_N \in B(\bP, \epsilon) \cap  \XN \in T_{N, \delta}(\bP) \right)\\ & \geq  - \int_{\Ac} \left(\ERS[\bPx | \Poisson] -1\right) dx - 1.
\end{split}
\end{equation}
\end{prop}
\begin{proof}[Proof of Proposition \ref{prop:quasicontinu}.]
We may assume that $\Ac$ is compact and that the intensity measure of $\bP$, denoted by $\brho$, is continuous, compactly supported and bounded below. Indeed we can always approximate $\bP$ by random point processes satisfying these additional assumptions. For any $N \geq 1$, we let $\brho_N(x) := \brho(x N^{-1/d})$ and we let $\Ac_N := N^{1/d} \Ac$.

In fact, for simplicity we will assume that $\Ac$ is some large hypercube. The argument below readily extends to the case where $\Ac$ can be tiled by small hypercubes, and any $C^1$ domain can be tiled by small hypercubes up to some “boundary parts” which are negligible for our concerns (a precise argument is given e.g. in \cite[Section 6]{LebSer}). 
\medskip

For $R > 0$, we let $\{ \tK_i \}_{i \in I}$ be a partition of $\Ac_N$ by hypercubes of sidelength $R$. For $R, M$, we denote by $\bP_{R, M}$ the restriction\footnote{That is, $\bP_{R, M}\in \overline{\mathcal M}(\Ac\times \config[\carr_R]).$} to $\carr_R$ of $\bP$, conditioned to the event 
\begin{equation} \label{conditioning}
\left\lbrace \left|\C \cap \carr_R\right| \leq MR^d\right\rbrace.
\end{equation}

\medskip

\textbf{Step 1.} \textit{Generating microstates.} \ \\
For any  $\epsilon  > 0$, for any $M, R > 0$, for any $\nu > 0$, for any $N \geq 1$, there exists a family $\Aabs = \Aabs(\epsilon, M, R, \nu, N)$ of point configurations $\Cabs$ such that: 
\begin{enumerate}
\item $\Cabs = \sum_{i \in I} \Cabs_i$ where $\Cabs_i$ is a point configuration in $\tK_i$.
\item $| \Cabs | = N$.
\item The “discretized” empirical process is close to $\bP_{R, M}$
\begin{equation}
\label{bPdbelo} \bPd(\Cabs)  :=  \frac{1}{|I|} \sum_{i \in I} \delta_{(N^{-1/d} x_i, \,\theta_{x_i} \cdot \Cabs_i)} \text{ belongs to } B(\bP_{R, M}, \nu),
\end{equation}
where $x_i$ denotes the center of $\tK_i$.
\item The associated empirical process is close to $\bP$
\begin{equation}
\label{bPcbelo} \bPc(\Cabs)  :=  \int_{\Ac} \delta_{(x,\, \theta_{N^{1/d}x} \cdot \Cabs)} \, dx \text{ belongs to } B(\bP, \epsilon).
\end{equation}
Note that $\bPc(\Cabs) =\bEmp_N( N^{-1/d}\Cabs) $.
\item  The volume of $\Aabs$ satisfies, for any $\epsilon > 0$
\begin{equation} \label{bonvolume}
 \liminf_{M \ti} \liminf_{R \ti} \frac{1}{R^d} \lim_{\nu \to 0}  \lim_{N \ti} \frac{1}{|I|} \log \Leb_{\Ac_N^N} \left( \Aabs \right) \geq - \int_{\Ac} \left(\ERS[\bPx | \Poisson] -1\right) - 1.
\end{equation}
\end{enumerate}
This is essentially \cite[Lemma 6.3]{LebSer} with minor modifications (e.g. the Lebesgue measure in \cite{LebSer} is normalized, which yields an additional logarithmic factor in the formulas).

We will make the following assumption on $\Aabs$
\begin{equation} \label{conditioning2}
|\Cabs_i| \leq 2MR^d \text{ for all } i \in I.
\end{equation}
Indeed for fixed $M$, when $\bPd$ is close to $\bP_{R,M}$ (for which \eqref{conditioning} holds), the fraction of hypercubes on which \eqref{conditioning2} fails to hold as well as the ratio of excess points over the total number of points (namely $N$) are both small. We may then “redistribute” these excess points among the other hypercubes without affecting \eqref{bPcbelo} and changing the energy estimates below only by a negligible quantity.
\medskip

\textbf{Step 2.} \textit{First energy estimate.}  \ \\
For any $R, M, \tau > 0$, the map defined by
\begin{equation*}
\Cabs\in\config(\carr_R) : \,  \longrightarrow  \Int_{\tau}[\Cabs,\Cabs] \wedge \frac{(2MR^d)^2}{\tau^s}
\end{equation*}
(where $\Int_{\tau}$ is as in \eqref{def:Inttau}) is continuous on $\config(\carr_R)$ and \textit{bounded} (this is precisely the reason for conditioning  that  the number of points are bounded). We may thus write, in view of \eqref{conditioning} and \eqref{bPdbelo}, \eqref{conditioning2},
\begin{multline*}
\int_{\Ac \times \config(\carr_R)} \Int_{\tau}\,  d\bPd = \int_{\Ac \times \config(\carr_R)} \Int_{\tau} \wedge \frac{(2MR^d)^2}{\tau^s} d\bPd \\
= \int_{\Ac \times \config(\carr_R)} \Int_{\tau} \wedge \frac{(2MR^d)^2}{\tau^s} d\bP_{R, M} + o_{\nu}(1) =  \int_{\Ac \times \config(\carr_R)} \Int_{\tau}\,  d\bP_{R, M} + o_{\nu}(1).
\end{multline*}
Moreover we have
\begin{equation*}
\lim_{M \ti} \lim_{R \ti} \frac{1}{R^d} \int_{\Ac \times \config(\carr_R)} \Int_{\tau} d\bP_{R, M} = \bWsP(\bP) + o_{\tau}(1), 
\end{equation*}
thus we see that, with \eqref{bPdbelo}
\begin{equation} \label{energietronqueeconverge}
\lim_{M \ti, R \ti} \lim_{\nu \to 0} \lim_{N \ti} \frac{1}{N} \sum_{i \in I} \Int_{\tau}[\Cabs_i, \Cabs_i] = \bWsP(\bP) + o_{\tau}(1).
\end{equation}
\medskip

\textbf{Step 3.} \textit{Regularization.} \ \\
In order to deal with the short-scale interactions that are not captured in $\Int_{\tau}$, we apply the regularization procedure of \cite[Lemma 5.11]{LebSer}. Let us briefly present this procedure:
\begin{enumerate}
\item We partition $\Ac_N$ by small hypercubes of sidelength $6\tau$.
\item If one of these hypercubes $\mc{K}$ contains more than one point, or if it contains a point and  one of the adjacent hypercubes also contains a point, we replace the point configuration in $\mc{K}$ by one with the same number of points but confined in the central, smaller hypercube $\mc{K}' \subset \mc{K}$ of side length $3 \tau$ and that lives on a lattice (the spacing of the lattice depends on the initial number of points in $\mc{K}$). 
\end{enumerate}
This allows us to control the difference $\Int - \Int_{\tau}$ in terms of the number of points in the modified hypercubes. 

In particular we replace $\Aabs$ by a new family of point configurations, such that
\begin{equation} \label{truncationerror}
\frac{1}{N} \sum_{i \in I} \left( \Int - \Int_{\tau} \right) [\C_i, \C_i] \leq C \tau^{-s-d}\Esp_{\bPd} \left[ \left( \left( \left|\C \cap \carr_{12\tau} \right| \right)^{2+s/d} - 1 \right)_+ \right].
\end{equation}
The right-hand side of \eqref{truncationerror} should be understood as follows: any group of points which were too close to each other (without any precise control) have been replaced by a group of points with the same cardinality, but whose interaction energy is now similar to that of a lattice. The energy of $n$ points in a lattice of spacing $\frac{\tau}{n^{1/d}}$ scales like $n^{2+ s/d} \tau^{-s}$, and taking the average over all small hypercubes, is similar to computing $\frac{1}{\tau^d} \Esp_{\bP_d}$.

As $\nu \to 0$ we may then compare the right-hand side of \eqref{truncationerror} with the same quantity for $\bP$, namely
\begin{equation*}
\tau^{-s-d} \Esp_{\bP} \left[ \left( \left(\left|\C \cap \carr_{12\tau} \right| \right)^{2+s/d}  - 1 \right)_+ \right]
\end{equation*}
which can be shown to be $o_{\tau}(1)$ (following the argument of \cite[Section 6.3.3]{LebSer}), because it is in turn of the same order as 
$$
\Esp_{\bP} \left[ \left( \Int - \Int_{\tau} \right) [K_1, K_1]  \right],
$$ 
which goes to zero as $\tau \to 0$ by dominated convergence.

We obtain 
\begin{equation} \label{errtrunpetite}
\lim_{\tau \to 0} \lim_{M, R \ti} \lim_{\nu \to 0} \frac{1}{N} \sum_{i \in I} \left( \Int - \Int_{\tau} \right) [\C_i, \C_i] =0
\end{equation}
and combining \eqref{errtrunpetite} with \eqref{energietronqueeconverge} we get that
\begin{equation} \label{avantscaling}
\lim_{\tau \to 0} \lim_{M \ti, R \ti} \lim_{\nu \to 0} \lim_{N \ti} \frac{1}{N} \sum_{i \in I} \Int[\C_i, \C_i] \leq \bWsP(\bP).
\end{equation}

\textbf{Step 4.} \textit{Shrinking the configurations.}  \edb{
This procedure is  borrowed from \cite{HSAdv}. It rescales the configuration by a factor less than one (but very close to $1$) effectively shrinking it and creating an empty boundary layer around each cube. Thus points belonging to different cubes are sufficiently well-separated so that interactions between the cubes are negligible--a much simpler approach to screening than that in the long range case.}

For $R > 0$ we let $\Rp := R^{\sqrt{d/s}}$.

It is not true in general that $\Int[\C, \C]$ can be split as the sum $\sum_{i \in I} \Int[\C_i, \C_i]$. However since the Riesz interaction decays fast at infinity it is approximately true if the configurations $\C_i$ are separated by a large enough distance. To ensure that, we “shrink” every configuration $\C_i$ in $\tK_i$, namely we rescale them by a factor $1 - \frac{\Rp}{R}$. This operation affects the discrete average \eqref{bPdbelo} but not the empirical process; i.e., for any $\epsilon > 0$, if $M, R$ are large enough and $\nu$ small enough, we may still assume that \eqref{bPcbelo} holds. The interaction energy in each hypercube $\tK_i$ is multiplied by $\left(1- \frac{\Rp}{R}\right)^{-s} = 1 + o_R(1)$, but the configurations in two distinct hypercubes are now separated by a distance at least $\Rp$. Since \eqref{conditioning2} holds, an elementary computation implies that we have, for any $i$ in $I$
\begin{equation*}
\Int[\C_i, \sum_{j \neq i} \C_j] = M^2 R^{d} \frac{R^d}{\Rp^s} O(1),
\end{equation*}
with a $O(1)$ depending only on $d,s$. We thus get 
\begin{equation*}
\Int[\C, \C] = \sum_{i \in I} \Int[\C_i, \C_i] +  N M^2 \frac{R^d}{\Rp^s} O(1),
\end{equation*}
but $\frac{R^d}{\Rp^s} = o_R(1)$ by the choice of $\Rp$ (and the fact that $d < s$) and thus (in view of \eqref{avantscaling} and the effect of the scaling on the energy)
\begin{equation}
\lim_{\tau \to 0} \lim_{M \ti, R \ti} \lim_{\nu \to 0} \lim_{N \ti} \frac{1}{N} \Int[\C, \C] \leq \bWsP(\bP).
\end{equation}
We have thus constructed a large enough (see \eqref{bonvolume}) volume of point configurations in $\Ac_N$ whose associated empirical processes converge to $\bP$ and such that
\begin{equation*}
\frac{1}{N} \Int[\C, \C] \leq \bWsP(\bP) + o(1).
\end{equation*}
We may view these configurations at the original scale by applying a homothety of factor $N^{-1/d}$, this way we obtain point configurations $\XN$ in $\Ac$ such that 
\begin{equation*}
\frac{1}{N^{1+s/d}} \Es(\XN) \leq \bWsP(\bP)+ o(1).
\end{equation*}
It is not hard to see that the associated empirical measure $\mu_N$ converges to the intensity measure of $\bP$ and since $V$ is continuous we also have
\begin{equation*}
\frac{1}{N} \int_{\R} V d\mu_N = \bV(\bP) + o(1).
\end{equation*}
This concludes the proof of Proposition \ref{prop:quasicontinu}.

\end{proof}

We may now prove the LDP lower bound.
\begin{proof}[Proof of Proposition \ref{prop:LDPLB}.]
Proposition \ref{prop:quasicontinu} implies \eqref{LDPLB}, indeed we have
\begin{multline*}
\fPNbeta ( B(\bP, \epsilon)) = \frac{1}{\ZNbeta} \int_{\Ac^N \cap \{\bEmp_N(\XN) \in B(\bP, \epsilon)\}} \exp\left(-\beta N^{-s/d} \HN(\XN)\right) d\XN \\ \geq \frac{1}{\ZNbeta} \int_{\Ac^N \cap  \{\bEmp_N(\XN) \in B(\bP, \epsilon) \cap T_{N, \delta}(\bP)\}} \exp\left(-\beta N^{-s/d} \HN(\XN)\right) d\XN \\
\geq \frac{1}{\ZNbeta} \exp \left(-\beta N \left(\fbarbeta(\bP) + \delta\right) \right) \int_{\Ac^N \cap \{\bEmp_N(\XN) \in B(\bP, \epsilon) \cap T_{N, \delta}(\bP)\}} d\XN,
\end{multline*}
and \eqref{quasicontinu} allows us to bound below the last integral as
\begin{equation*}
\liminf_{\delta \to 0, \epsilon \to 0, N \ti} \frac{1}{N} \log \int_{\Ac^N \cap \{\bEmp_N(\XN) \in B(\bP, \epsilon) \cap T_{N, \delta}(\bP)\}} d\XN \geq -\int_{\Ac} \left(\ERS[\bPx | \Poisson] -1\right) - 1.
\end{equation*}
\end{proof}

\subsection{Proof of Theorem \ref{theo:LDPemp} and Corollary \ref{coro:ZNbeta}}
From Proposition \ref{prop:LDPUB} and Proposition \ref{prop:LDPLB}, the proof of Theorem \ref{theo:LDPemp} is standard. Exponential tightness of $\fPNbeta$ comes for free (see e.g. \cite[Section 4.1]{LebSer}) because the average number of points is fixed, and we may thus improve the weak large deviations estimates \eqref{LDPUB}, \eqref{LDPUB} into the following: for any $A \subset \psunbconfig$ we have
\begin{multline*}
- \inf_{\mathring{A}} \fbarbeta + \liminf_{N \ti} \left(- \frac{\log \ZNbeta}{N}\right)  \\ \leq \liminf_{N \ti}\frac{1}{N} \log \fPNbeta(A) \leq \limsup_{N \ti}\frac{1}{N} \log \fPNbeta(A) \\ \leq - \inf_{\overline{A}} \fbarbeta \limsup_{N \ti} \left(- \frac{\log \ZNbeta}{N}\right).
\end{multline*}
We easily deduce that
\begin{equation*}
\lim_{N \ti} \frac{\log \ZNbeta}{N} = - \min_{\psunbconfig} \fbarbeta, 
\end{equation*}
which proves Corollary \ref{coro:ZNbeta}, and that the LDP for $\fPNbeta$ holds as stated in Theorem \ref{theo:LDPemp}.

\subsection{Proof of Theorem \ref{theo:LDPmesure}}
\label{sec:contractionprinciple}
\begin{proof}
Theorem \ref{theo:LDPmesure} follows from an application of the “contraction principle” (see e.g. \cite[Section 3.1]{MR3309619}). Let us consider the map $\pbconfig \to \probas(\Ac)$ defined by
\begin{equation*}
\tIntens : \bP \mapsto \int_{\Ac} \delta_x \Esp_{\bPx}[\C \cap \carr_1].
\end{equation*}
It is continuous on $\psbconfig$ and coincides with $\bIntens$. By the contraction principle, the law of $\tIntens(\bEmp(\XN))$ obeys a large deviation principle governed by 
\begin{equation*}
\rho \mapsto \inf_{\bIntens(\bP) = \rho} \fbarbeta(\bP), 
\end{equation*}
which is easily seen to be equal to $\Ibeta(\rho)$ as defined in \eqref{def:IbetaA}.

For technical reasons (a boundary effect), it is not true in general that $\tIntens(\bEmp(\XN)) = \emp(\XN)$, however we have
\begin{equation*}
\dist_{\probas(\Ac)} \left(\tIntens(\bEmp(\XN)),  \bEmp(\XN)\right) = o_N(1),
\end{equation*}
uniformly for $\XN \in \Ac$. In particular, the laws of $\tIntens(\bEmp(\XN))$ and of $\emp(\XN)$ are exponentially equivalent (in the language of large deviations), thus any LDP can be transferred from one to the other. This proves Theorem \ref{theo:LDPmesure}.
\end{proof}

\section{Additional proofs: Propositions \ref{prop:muVbeta}, \ref{prop:minimizers} and \ref{prop:crystallization1d}}
\label{sec:addproofs}
\subsection{Limit of the empirical measure} \label{sec:LDPempir}
From Theorem \ref{theo:LDPmesure} and the fact that $\Ibeta$ is strictly convex we deduce that  $\emp(\XN)$ converges almost surely to the unique minimizer of $\Ibeta$. 
\begin{proof}[Proof of Proposition \ref{prop:muVbeta}.]  \ \\
 First, if $V = 0$ and $\Omega$ is bounded, $\Ibeta$ can be written as
\begin{multline*}
 \Ibeta(\rho) :=  \int_{\Ac} \rho(x) \inf_{P \in \psunconfig} \left( \beta \rho(x)^{s/d} \WsP(P) +\ERS[P|\Poisson] \right)dx \\ + \int_{\Ac}  \rho(x) \log \rho(x) \, dx.
\end{multline*}
We claim that both terms in the right-hand side are minimized when $\rho$ is the uniform probability measure on $\Ac$ (we may assume $|\Ac| = 1$ to simplify, without loss of generality). This property is well-known for the relative entropy term $\int_{\Ac} \rho \log \rho$, and we now prove it for the energy term. 
\ed{First, let us  observe that 
$$
\alpha \mapsto \inf_{P \in \psunconfig} \left( \beta \alpha^{1+s/d} \WsP(P) + \alpha \ERS[P|\Poisson]\right)
$$
is convex in $\alpha$ \edb{since it is the infimum over a family of convex functions  (recall } that $\alpha \mapsto \alpha^{1+s/d}$ is convex in $\alpha$ and that $\WsP$ is always positive). Since $|\Ac| = 1$ we have, by Jensen's  inequality, 
\begin{multline*}
\int_{\Ac}  \inf_{P \in \psunconfig} \left( \beta \rho(x)^{1+s/d} \WsP(P) + \rho(x)\ERS[P|\Poisson] \right)dx 
\\
\geq \inf_{P \in \psunconfig} \left( \beta \left(\int_{\Ac} \rho(x) \right)^{1+s/d} \WsP(P)  + \left(\int_{\Ac} \rho(x)\right) \ERS[P|\Poisson] \right),
\end{multline*}
and since $\int_{\Ac} \rho = 1$, we conclude that $\Ibeta$ is minimal for $\rho \equiv 1$.} Thus the empirical measure converges almost surely to the uniform probability measure on $\Ac$, which proves the first point of Proposition \ref{prop:muVbeta}. 

 Next, let us assume  that $V$ is arbitrary and $\Omega$ bounded. It is not hard to see that for the minimizer $\mu_{V, \beta}$ of $\Ibeta$ we have, as $\beta \to 0$.
\begin{equation*}
\Ibeta(\mu_{V, \beta}) \geq \Ibeta(\rhounif) + O(\beta),  
\end{equation*}
where $\rhounif$ is the uniform probability measure on $\Ac$. Moreover it is also true (as proven above) that the first term in the definition of $\Ibeta$ is minimal for $\rho = \rhounif$. We thus get that, as $\beta \to 0$
$$
\int_{\Ac} \mu_{V, \beta} \log \mu_{V, \beta} - \int_{\Ac} \rhounif \log \rhounif = O(\beta),
$$
\ed{in other words the relative entropy of $\mu_{V, \beta}$ with respect to $\rhounif$ converges to $0$ as $\beta \to 0$. The Csiszár-Kullback-Pinsker's inequality allows us to bound the square of the total variation distance between $\mu_{V, \beta}$ and $\rhounif$ by the relative entropy (up to a multiplicative constant), and thus $\mu_{V, \beta}$ converges (in total variation) to the uniform probability measure on $\Ac$ as $\beta \to 0$.} This proves the second point of Proposition \ref{prop:muVbeta}.

 Finally for $V$ arbitrary, the problem of minimizing of $\Ibeta$ is, as $\beta \ti$, similar to minimizing
$$
\beta \left( \int_{\Ac} \rho(x)^{1+s/d} \min \WsP  dx + \int_{\Ac} \rho(x) V(x)  dx \right).
$$
Since $\min \WsP = \Csd$ we recover (up to a multiplicative constant $\beta > 0$) the minimization problem studied in \cite{Hardin:2016kq}, namely the problem of minimizing
$$
\Csd \int_{\Ac} \rho(x)^{1+s/d}  dx + \int_{\Ac} \rho(x) V(x)  dx,
$$
among probability densities, whose (unique) solution is given by $\mu_{V, \infty}$. 

In order to prove that $\mu_{V, \beta}$ converges to $\mu_{V, \infty}$ as $\beta \ti$, we need to make that heuristic rigorous, which requires an adaptation of \cite[Section 7.3, Step 2]{leble2016logarithmic}. We claim that there exists a sequence $\{P_k\}_{k \geq 1}$ in $\psunconfig$ such that
\begin{equation} \label{Pkapprox}
\lim_{k \ti} \WsP(P_k) = \Csd, \quad \forall k \geq 1,  \ERS[P_k|\Pi] < + \infty.
\end{equation}
We could think of taking $P_k = P$ where $P$ is some minimizer of $\WsP$ among $\psunconfig$, but it might have infinite entropy (e.g., if $P$ was the law of the stationary process associated to a lattice, as in dimension $1$). We thus need to “expand” $P$ (e.g., by making all the points vibrate independently in small balls as described in \cite[Section 7.3, Step 2]{leble2016logarithmic} in the case of the one-dimensional lattice). We may then write that, for any $\beta > 0$ and $k \geq 1$, 
\begin{multline*}
\Ibeta(\mu_{V, \beta}) \leq \Ibeta(\mu_{V, \infty}) \leq  \beta \left( \int_{\Ac} \mu_{V, \infty}(x)^{1+s/d} \WsP(P_k) + \int_{\Ac} \mu_{V, \infty}(x) V(x) dx \right) \\ + \ERS[P_k|\Pi] + \int_{\Ac} \mu_{V, \infty}(x) \log \mu_{V, \infty}(x) \\
\leq \beta \left(\int_{\Ac} \mu_{V, \infty}(x)^{1+s/d} \Csd + \int_{\Ac} \mu_{V, \infty}(x) V(x) dx \right) + \ERS[P_k|\Pi] + \beta o_k(1),
\end{multline*}
where we have used \eqref{Pkapprox} in the last inequality. Choosing $\beta$ and $k$ properly \ed{so that $k \to \infty$ as $\beta \to \infty$, while assuring that the $\beta o_k(1)$ term goes to zero, we have}
\begin{multline*}
\Csd \int_{\Ac} \mu_{V, \beta}(x)^{1+s/d} + \int_{\Ac} \mu_{V, \beta}(x) V(x) \leq \Csd \int_{\Ac} \mu_{V, \infty}(x)^{1+s/d}  dx + \int_{\Ac} \mu_{V, \infty}(x) V(x)  dx \\ 
+ o_{\beta \ti }(1).
\end{multline*}
By convexity, it implies that $\mu_{V, \beta}$ converges to $\mu_{V, \infty}$ as $\beta \ti$.
\end{proof}

\subsection{The case of minimizers}
\begin{proof}[Proof of Proposition \ref{prop:minimizers}]
Let $\{\XN\}_N$ be a sequence of $N$-point configuration such that for all $N \geq 1$, $\XN$ minimizes $\HN$. From Proposition \ref{prop:glinf} we know that (up to extraction), $\{\bEmp(\XN)\}_N$ converges to some $\bP \in \psunbconfig$ such that
\begin{equation} \label{minimlsci}
\bWsP(\bP) + \bV(\bP) \leq \liminf_{N \ti} \frac{\HN(\XN)}{N^{1+s/d}},
\end{equation}
and we have, by \eqref{Wsfatou}, \eqref{minA1} and the scaling properties of $\WsP$
\begin{equation} \label{minimplusgrand}
\bWsP(\bP) + \bV(\bP) \geq \Csd \int_{\Ac} \rho(x)^{1+s/d} dx + \int_{\Ac} V(x) \rho(x) dx.
\end{equation}
where $\rho = \bIntens(\bP)$. We also know that the empirical measure $\emp(\XN)$ converges to the intensity measure $\rho = \bIntens(\bP)$. 

On the other hand, from \cite[Theorem 2.1]{Hardin:2016kq} we know that $\emp(\XN)$ converges to some measure $\mu_{V, \infty}$ which is defined as follows: define $L$ to be the unique solution of 
 \begin{equation*}
 \int_{\Ac} \left[\frac{L-V(x)}{\Csd(1+s/d)}\right]_+^{d/s} dx=1.  
 \end{equation*}
 and let then $\mu_{V, \infty}$ be given by
  \begin{equation}\label{muVinfty}
\mu_{V, \infty}(x):= \left[\frac{L-V(x)}{\Csd (1+s/d)}\right]_+^{d/s} \qquad (x\in  \Ac).
 \end{equation}
It is proven in \cite{Hardin:2016kq} that $\mu_{V, \infty}$ minimizes the quantity 
 \begin{equation}  \label{queminimim}
 \Csd \int_{\Ac} \rho(x)^{1+s/d}\, dx + \int V(x) \rho(x)\, dx,
 \end{equation}
among all probability density functions $\rho$ supported on $\Ac$. It is also proven that 
\begin{equation} \label{minimizerminimize}
\lim_{N \ti} \frac{\HN(\XN)}{N^{1+s/d}} = \Csd \int_{\Ac} \mu_{V, \infty}(x)^{1+s/d}\, dx + \int V(x) \mu_{V, \infty}(x)\, dx,
\end{equation}

By unicity of the limit we have $\rho := \bIntens(\bP) = \mu_{V, \infty}$. In view of \eqref{minimlsci}, \eqref{minimplusgrand}, \eqref{minimizerminimize} and by the fact that $\mu_{V, \infty}$ minimizes \eqref{queminimim} we get that
$$
\bWsP(\bP) + \bV(\bP) = \Csd \int_{\Ac} \mu_{V, \infty}(x)^{1+s/d} dx + \int_{\Ac} V(x) \mu_{V, \infty}(x) dx,
$$
and that $\bP$ is in fact a minimizer of $\bWsP + \bV$. We must also have 
$$\bWsP(\bP) = \Csd \int_{\Ac} \mu_{V, \infty}(x)^{1+s/d}\, dx$$
hence (in view of \eqref{Wsfatou}) we get 
\begin{equation*}
\Ws(\C) = \Csd \mu_{V, \infty}(x)^{1+s/d} = \min_{\config_{\mu_{V, \infty}(x)}} \Ws, \text{ for $\bP$-a.e. $(x,\C)$,}
\end{equation*}
which concludes the proof.
\end{proof}

\subsection{The one-dimensional case}
Proposition \ref{prop:crystallization1d} is very similar to the first statement of \cite[Theorem 3]{leble2016logarithmic}, and we sketch its proof here. 
\begin{proof}[Proof of Proposition \ref{prop:crystallization1d}]
First, we use the expression of $\WsP$ in terms of the two-point correlation function, as presented in \eqref{Wrho2P2}
$$
\WsP(P) = \liminf_{R \ti} \int_{[-R, R]^d} \frac{1}{|v|^s} \rho_{2,P}(v) \left( 1 - \frac{|v|}{R} \right) dv.
$$
Then, we split $\rho_{2, P}$ as the sum 
$$
\rho_{2, P} = \sum_{k=1}^{+\infty} \rho_{2,P}^{(k)},
$$
where $\rho_{2,P}^{(k)}$ is the correlation function of the $k$-th neighbor (which makes sense only in dimension $1$). It is not hard to check that
$$\int \rho_{2,P}^{(k)}(x) = 1 \text{ and } \int x \rho_{2,P}^{(k)}(x) = k$$
(the last identity holds because $P$ has intensity $1$ and is stationary). Using the convexity of 
$$
v \mapsto \frac{1}{|v|^s}  \left( 1 - \frac{|v|}{R} \right), 
$$
we obtain that for any $k \geq 1$ it holds
$$
\int  \frac{1}{|v|^s}  \left( 1 - \frac{|v|}{R} \right) \rho_{2,P}^{(k)} dv \geq \int \frac{1}{|v|^s}  \left( 1 - \frac{|v|}{R} \right) \delta_{k}(v) dv = \int \frac{1}{|v|^s}  \left( 1 - \frac{|v|}{R} \right) \rho_{2,P_{\Z}}^{(k)}(v) dv,
$$
where $P_{\Z} = u + \Z$ (with $u$ uniform in $[0,1]$) thus we have
$$
\WsP(P) \geq \WsP(P_{\Z}), 
$$
which proves that $\WsP$ is minimal at $P_{\Z}$.
\end{proof}

\noindent
{\bf Acknowledgements:}  The authors thank the referee for a careful reading and helpful suggestions.

 \bibliographystyle{alpha}

\bibliography{LDPHyper_bibfile}

 \end{document}